\documentclass[a4paper,10pt]{article}

\usepackage[utf8]{inputenc}
\usepackage[T1]{fontenc}
\usepackage[english]{babel}
\usepackage{fullpage}
\usepackage{microtype}
\usepackage{xspace}
\usepackage{hyperref}
\usepackage{amsmath,amssymb,amsfonts,amsthm}
\usepackage{csquotes}

\usepackage{lmodern}

\usepackage{authblk}

\usepackage{xcolor}
\usepackage{afterpage}
\usepackage{longtable}
\usepackage[font=normalsize,labelfont=bf]{caption}

\usepackage{colortbl}
\usepackage{algorithm}
\usepackage[noend]{algpseudocode}
\algdef{SE}[DOWHILE]{Do}{doWhile}{\algorithmicdo}[1]{\algorithmicwhile\ #1}%
\newcommand\CONDITION[2]%
  {\begin{tabular}[t]{@{}l@{}l@{}}
     #1&#2
   \end{tabular}%
  }
\algdef{SE}[WHILE]{While}{EndWhile}[1]%
  {\algorithmicwhile\ \CONDITION{#1}{\ \algorithmicdo}}%
  {\algorithmicend\ \algorithmicwhile}
\algdef{SE}[FOR]{For}{EndFor}[1]%
  {\algorithmicfor\ \CONDITION{#1}{\ \algorithmicdo}}%
  {\algorithmicend\ \algorithmicfor}
\algdef{S}[FOR]{ForAll}[1]%
  {\algorithmicforall\ \CONDITION{#1}{\ \algorithmicdo}}
\algdef{SE}[REPEAT]{Repeat}{Until}{\algorithmicrepeat}[1]%
  {\algorithmicuntil\ \CONDITION{#1}{}}
\algdef{SE}[IF]{If}{EndIf}[1]%
  {\algorithmicif\ \CONDITION{#1}{\ \algorithmicthen}}%
  {\algorithmicend\ \algorithmicif}%
\algdef{C}[IF]{IF}{ElsIf}[1]%
  {\algorithmicelse\ \algorithmicif\ \CONDITION{#1}{\ \algorithmicthen}}%

\usepackage{enumitem}
\usepackage{xpatch}
\xpretocmd{\footnote}{\unskip}{}{}

\usepackage{mathtools}
\allowdisplaybreaks  

\usepackage{todonotes}

\newtheorem{definition}{Definition}[section]

\newtheorem{lemma}[definition]{Lemma}
\newtheorem{claim}[definition]{Claim}
\newtheorem{theorem}[definition]{Theorem}

\usepackage{xspace}
\newcommand{\node}{agent\xspace}
\newcommand{\nodes}{agents\xspace}
\newcommand{\emp}{empirical\xspace}

\newcommand{\consensus}{almost-consensus\xspace}
\newcommand{\Consensus}{Almost-consensus\xspace}
\newcommand{\choice}{2-Choices dynamics\xspace}
\newcommand{\advchoiced}[0]{Biased-2-Choices dynamics\xspace}
\newcommand{\advchoice}[2]{Biased-2-Choices(\text{#1, #2})\xspace}
\newcommand{\cond}{\ \middle| \ }
\newcommand{\core}{\mathcal{C}}
\newcommand{\periphery}{\mathcal{P}}

\newcommand{\vol}{\ensuremath\mathrm{vol}}

\newcommand{\gp}{\ensuremath g_p}
\newcommand{\fp}{\ensuremath f_p}

\newcommand{\PP}[2]{\Prob{#1}{#2}}
\newcommand{\Prob}[2]{\mathbf{P}_{#1} \left( #2 \right)}

\newcommand{\Ex}[1]{\mathbf{E}\left[ #1 \right]}
\newcommand{\Var}[1]{\mathbf{Var}\left[ #1 \right]}

\newcommand{\abs}[1]{\vert #1 \vert}

\newcommand{\dijcup}{\mathbin{\dot{\cup}}}

\newcommand{\poly}{ {\mathrm{poly}}}

\newcommand{\bigO}{\mathcal{O}}

\renewcommand{\leq}{\leqslant}
\renewcommand{\geq}{\geqslant}

\renewcommand{\epsilon}{\varepsilon}

\newcommand{\D}{3 - 2\sqrt{2}}

\begin{document}
\title{Phase Transition of the 2-Choices Dynamics\\on Core-Periphery Networks\thanks{Work partially carried out at Max Planck Institute for Informatics --- Saarbr\"ucken, Germany.}}

\author{Emilio Cruciani}
\author{Emanuele Natale}
\affil{%
 Inria, I3S, CNRS, UCA --- Sophia Antipolis, France
 \authorcr\texttt{emilio.cruciani@inria.fr},
 \authorcr\texttt{natale@i3s.unice.fr}
}

\author{\\Andr\'e Nusser}
\affil{%
 Max Planck Institute for Informatics --- Saarbr\"ucken, Germany
 \authorcr\texttt{anusser@mpi-inf.mpg.de}
}

\author{Giacomo Scornavacca}
\affil{%
 Università di Sassari --- Sassari, Italy
 \authorcr\texttt{giacomo.scornavacca@gmail.com}
}
\date{}

\maketitle

\begin{abstract}
 Consider the following process on a network: Each agent initially holds either
 opinion \emph{blue} or \emph{red}; then, in each round, each agent looks at two random
 neighbors and, if the two have the same opinion, the agent adopts it.
 This process is known as the \emph{2-Choices} dynamics and is arguably the
 most basic non-trivial \emph{opinion dynamics} modeling voting behavior on
 social networks.
 Despite its apparent simplicity, 2-Choices has been analytically
 characterized only on restricted network classes---under
 assumptions on the initial configuration that establish it as a fast
 \emph{majority consensus} protocol.

 In this work, we aim at contributing to the understanding of the 2-Choices dynamics by considering
 its behavior on a class of networks with core-periphery structure, a
 well-known topological assumption in social networks.
 In a nutshell, assume that a densely-connected subset of agents, the \emph{core}, holds
 a different opinion from the rest of the network, the
 \emph{periphery}. Then, depending on the strength of the cut between the
 core and the periphery, a phase-transition phenomenon occurs: Either the core's opinion
 rapidly spreads among the rest of the network, or a
 \emph{metastability} phase takes place, in which both opinions coexist in
 the network for superpolynomial time.
 The interest of our result is twofold. On the one hand, by looking at the
 2-Choices dynamics as a simplistic model of competition among opinions in
 social networks, our theorem sheds light on the \emph{influence} of the
 core on the rest of the network, as a function of the core's connectivity
 toward the latter. On the other hand, we provide one of
 the first analytical results which shows a heterogeneous behavior of a simple 
 dynamics as a function of structural parameters of the network.
 Finally, we validate our theoretical predictions with extensive experiments on real
 networks.
\end{abstract}

\newpage

\section{Introduction}

\emph{Opinion dynamics} (for short, dynamics) are simplistic mathematical
models for the competition of agents' opinions on social networks%
~\cite{mossel_majority_2014}. 
In a nutshell, given a network where each agent initially supports an opinion (from a finite set), a dynamics is a simple rule which \nodes\
apply to update their opinion based on that of their neighbors. 
Proving theoretical results on dynamics is a challenging mathematical
endeavor which may require the development of new analytical techniques%
~\cite{benjamini_convergence_2016,becchetti_plurality_2015}.

As dynamics are aimed at modeling the spread of opinions, a central issue is
to understand under which conditions the network reaches a \emph{consensus}, i.e.,
a state where the whole network is taken over by
a single opinion~\cite{mossel_opinion_2017}.
In this respect, most efforts have been directed toward obtaining
\emph{topology-independent} results, which disregard the initial opinions'
placement on the network%
~\cite{becchetti_simple_2014,becchetti_stabilizing_2016,cooper_fast_2016}. 

The \emph{trivial} example of dynamics is the so-called
\emph{Voter Model}, in which in each round each agent copies the opinion of a
random neighbor. This classical model arises in physics and computer science and, despite its apparent simplicity,  
some properties have been proven only recently~\cite{liggett_stochastic_2013,hassin_distributed_2001,berenbrink_bounds_2016}.
The simplest non-trivial example is then arguably the \emph{2-Choices} dynamics, 
in which \nodes\ choose two random neighbors and switch to their
opinion only if they coincide~\cite{cooper_power_2014} (see Definition%
~\ref{def:2choice}).
Still, the analysis of the 2-Choices dynamics has been limited to networks with
good expansion properties, and the theoretical guarantees provided so far are
essentially independent from the positioning of initial opinions~\cite{cooper_fast_2015}.

In this work, we aim at contributing to the general understanding of the
evolution of simple opinion dynamics in richer classes of network topologies
by studying their behavior theoretically and
empirically on core-periphery networks. 
Core-periphery networks are typical economic and social networks which exhibit
a core-periphery structure, a well-known concept in the analysis of social networks in sociology and economics~\cite{borgatti_models_2000, easley_networks_2010}, which defines a bipartition of the \nodes\ into \emph{core} and \emph{periphery}, 
such that certain key features are identified. 

We consider an axiomatic framework that has been introduced in previous work in computer science~\cite{avin_core-periphery_2014}, which is based on two parameters only, \emph{dominance} and \emph{robustness}. The ranges for these parameters in the theorems we obtain include the values used in the the experimental part of this work, in which our results are validated on important datasets of real-world networks. 

Intuitively, the core is a set of \nodes\ that \emph{dominates} the rest of the
network. In order to do so, it maintains a large amount of external influence
on the periphery, higher than or at least comparable to the internal influence
that the periphery has on itself. Similarly, to maintain its \emph{robustness},
namely to hold its position and stick to its opinions, the core must be able to
resist the ``outside'' pressure in the form of external influence. To achieve
that, the core must maintain a higher (or at least not significantly lower)
influence on itself than the external influence exerted on it by the periphery. 
Both, high dominance and high robustness, are
essential for the core to be able to maintain its dominating status in the
network. 
Moreover, it seems natural for the core size to be as small as possible. 
In social-network terms this is motivated by the idea that if membership in a
social elite entails benefits, then keeping the elite as small as possible
increases the share for each of its members.

The above requirements are formalized in the following axioms~\cite{avin_core-periphery_2014}. 
Given a network $G=(V,E)$ and two subsets of \nodes\ $A,B\subseteq V$, let $e(A,B) :=
\{(u,v) \,|\, u\in A, v\in B,(u,v) \in E\}$ be the \emph{set of cut edges} between $A$ and $B$.
The \emph{density} of a set $X \subseteq V$ is defined as $\delta (X)=\abs{e(X, X)}/|X|$.
Let $c_d$ and $c_r$ be two positive constants
and let $V = \core \dijcup \periphery$, where $\core$ is the set of \nodes\ in the core and $\periphery$ the set of agents in the periphery. 
Then, the axioms are as follows:
\begin{itemize}
    \item \emph{Dominance}: The core’s influence dominates the periphery, i.e.,
        $\abs{e(\core, \periphery)} \geq c_d \cdot \abs{e(\periphery, \periphery)}$.
    \item \emph{Robustness}: The core can withstand outside influence from the
        periphery, i.e., $\abs{e(\core, \core)} \geq c_r \cdot
		\abs{e(\periphery, \core)}$.
    \item \emph{Compactness}: The core is a minimal set satisfying the above dominance and robustness
        axioms.\footnote{The core is a minimal set and not necessarily the minimum one.}
    \item \emph{Density}: The core is denser than the whole network, i.e.,
        $\delta (\core) > \delta (V)$.
\end{itemize}
Our analytical and experimental results leverage the dominance and
robustness axioms only (see Definition~\ref{def:regcore}), showing how assumptions
on the values of $c_d$ and $c_r$ are sufficient to provide a good characterization
of the behavior of the dynamics. 

We consider the \choice\ in core-periphery
networks when starting from \emph{natural} initial configurations in which the
core and the periphery have different opinions. 
Our experiments on real-world networks show that the execution of the \choice\
tends to fall mainly within two opposite kinds of possible behavior: 
\begin{itemize}
	\item \emph{Consensus:} The opinion of the core spreads in the periphery and takes over the network in a short time.
	\item \emph{Metastability:} The periphery \emph{resists} and, although the opinion of the core
        may continuously ``infect'' \nodes\ in the periphery, most of them
        retain the initial opinion. 
\end{itemize}
By comprehending the underlying principles which govern the aforementioned
phenomena, we aim at a twofold contribution:
\begin{itemize}
    \item We seek for the first results on
        basic non-trivial opinion dynamics, such as 2-Choices, in order to
        characterize its behavior:
		(i) on new classes of topologies other than networks with strong
		expansion;
		(ii) as a function of the process' initial configuration.

	\item We look for a \emph{dynamic} explanation for the
		axioms of core-periphery networks: By investigating the interplay
        between the core-periphery axioms and the evolution of simple 
        opinion dynamics, we want to get insights on \emph{dynamical
        properties} which are implicitly responsible for causing social and
        economic \nodes\ to form networks with a core-periphery structure. 
\end{itemize}

\subsection*{Original Contribution}

\newcommand{\coreind}{c}
\newcommand{\periind}{p}
In order to understand what network key properties are responsible for the
aforementioned dichotomy between a long \emph{metastable} and a fast
\emph{consensus} behavior, we theoretically investigate a
class of networks belonging to the core-periphery model.

To further simplify the theoretical analysis, in Theorem~\ref{thm:noisyd} we
initially consider the setting in which \emph{\nodes\ in the core are stubborn},
i.e., they don't change their initial opinion.
We later show, in Theorem~\ref{thm:phase_transition_cp}, how to substitute this assumption with
milder hypotheses on the core's structure. 
We remark that the evolution of the 2-Choices dynamics, together
with the latter assumption on the stubbornness of the core, can be regarded as
a SIS-like epidemic model~\cite{hethcote_mathematics_2000,
newman_spread_2002}. In such a model, the network is the subgraph induced by the
periphery and the infection probability is given by the 2-Choices dynamics,
which also determines a certain probability of \emph{spontaneous infection}
(that in the original process corresponds to the fact that \nodes\ in the
periphery interact with \nodes\ in the core). This interpretation of our
results may be of independent interest.

The common difficulty in analyzing opinion dynamics is the lack of general
tools which allow to rigorously handle their intrinsic nonlinearity and
stochastic dependencies~\cite{mossel_majority_2014, cooper_power_2014,
becchetti_stabilizing_2016, cooper_fast_2016, mossel_opinion_2017}. Hence, the
difficulty usually resides in identifying some crucial key quantities for which
ad-hoc analytical bounds on the expected evolution are derived. Our approach is
yet another instance of such efforts:
In Section~\ref{sec:theo} we provide a careful bound on the expected
change of the number of agents supporting a given opinion. Together with the
use of Chernoff bounds, we obtain a concentration of probability around the
expected evolution.
Rather surprisingly, our analysis on the concentrated probabilistic behavior
turns out to identify a \emph{phase transition} phenomenon: 
\begin{quote}
    There exists a universal constant $c^{\star} = \frac{\sqrt{2} - 1}{2}$ such that, on
	any core-periphery network of $n$ \nodes{}, if the dominance parameter $c_d$ is greater than
	$c^{\star}$, an \emph{\consensus}\ is rapidly reached,
	with high probability;
    \footnote{We further emphasize that our analysis is not only
        \emph{mean-field}. In addition to describing how the process evolves
        in expectation, we show that the process
        does not deviate significantly from how it behaves in expectation
        \emph{with high probability} (w.h.p.), i.e., with probability at least $1-\mathcal{O}(n^{-c})$
        for some constant $c>0$.}
	otherwise, if $c_d$ is
	less than $c^{\star}$, a \emph{metastable} phase in which the periphery retains
    its opinion takes place, namely, although the opinion
    of the core may continuously ``infect'' \nodes\ in the periphery, most of
    them will retain the initial opinion for superpolinomially-many rounds, with high
    probability.\footnote{Our technique does not allow a direct analysis of the case in which $c_d$ is
	exactly equal to $c^{\star}$; in that respect, we provide an  analysis in the simplified setting of Theorem \ref{thm:noisyd}.}
\end{quote}

We validate our theoretical predictions by extensive experiments on real-world
networks chosen from a variety of domains including social networks,
communication networks, road networks, and the web. We thoroughly discuss 
our results in Section~\ref{sec:experiments}, however, we briefly want to
highlight the key results of our experiments in the following.
The experiments showed some weaknesses of the core extraction heuristic used
in~\cite{avin_core-periphery_2014}. To avoid those drawbacks, we designed a new
core extraction heuristic which repeatedly calculates densest-subgraph approximations.
Our experimental results on real-world networks show a strong correlation with the theoretical
predictions made by our model. They further suggests an empirical threshold 
larger than $ c^{\star}$ for which the aforementioned correlation is even
stronger. 
We discuss which aspects of the current theoretical model may be responsible for
such discrepancy, and thus identify possibilities for a model which is even more accurate.

We remark that our investigation represents an original contribution with
respect to the line of research on \emph{consensus} discussed in Section%
~\ref{sec:related}, as it shows a drastic change in behavior for the 2-Choices
dynamics on an arguably \emph{typical} broad family of social networks which is
not directly characterized by expansion properties. 
In particular, the convergence to the core's opinion in our theoretical and
experimental results is a highly nontrivial fact when compared to previous
analytical works (see Section \ref{sec:experiments} for more details).

Overall, our theoretical and experimental results highlight new potential
relations between the typical core-periphery structure in social and economic networks
and the behavior of simple opinion dynamics---both, in terms of getting insights
into the driving forces that may determine certain structures to appear frequently
in real-world networks, as well as in terms of the possibility to provide
analytical predictions on the outcome of simplistic models of interaction in
networks of agents.

\subsection*{Roadmap}
In Section \ref{sec:related} we present related work on opinion dynamics and core-periphery networks.
In Section \ref{sec:theo} we analyze the \advchoiced{}, proving a phase transition on general dense graphs. We complete our results by analyzing the behaviour of the dynamics on the threshold of the phase transition.
In Section \ref{sec:theo-cp} we show how to translate the results of Theorem~\ref{thm:noisyd} in a formal analysis of the \choice\ on a family of Core-Periphery networks. 
In Section \ref{sec:experiments} we propose a new heuristic for extracting the Core from real networks and we present simulations of the \choice\ when the starting color configuration reflects our Core-Periphery partitioning.

\section{Related Work}
\label{sec:related}

Simple models of interaction between pairs of nodes in a network are 
studied since the first half of the 20th century in statistical mechanics
where mathematical models of ferromagnetism led to the study of Ising and Potts
models~\cite{liggett_interacting_2012}. 
A different perspective later came from diverse sciences such as economics and
sociology where averaging-based opinion dynamics such as the DeGroot model
were investigated~\cite{french_formal_1956, harary_criterion_1959,
degroot_reaching_1974, golub_naive_2010, jackson_social_2010}.

More recently, computer scientists have started to contribute to the
investigating of opinion dynamics for mainly two reasons. First, with the
advent of the Internet, huge amounts of of data from social
networks are now available.
As the law of large numbers often allows to assume crude simplifications on the
agents' behavior in such networks~\cite{easley_networks_2010},
investigating opinion dynamics allows for a more fine-grained understanding of the
evolution of such systems.
Second and somehow complementary to the previous motivation, technological systems of
computationally simple agents, such as mobile sensor networks, often require the
design of computationally primitive protocols. These protocols end up being
surprisingly similar, if not identical, to many opinion dynamics which emerge
from a simplistic mathematical modeling of agents' behavior in fields such as
sociology, biology, and economy%
~\cite{hassin_distributed_2001,becchetti_plurality_2015,mossel_opinion_2017}. 

A substantial line of work has recently been devoted to investigating the use
of simple opinion dynamics for solving the plurality consensus problem in
distributed computing. 
The goal in this problem setting is for each node to be aware of the most frequent 
initially supported opinion after a certain time~\cite{
doerr_stabilizing_2011, becchetti_simple_2014,
becchetti_plurality_2015, becchetti_stabilizing_2016,
ghaffari_polylogarithmic_2016, george_giakkoupis_efficient_2016,
colin_cooper_linear_2016, cooper_discordant_2016,
cooper_coalescing-branching_2016}.
The seminal work by Hassin and Peleg~\cite{hassin_distributed_2001}
introduced for the first time the study of a synchronous-time version of the
Voter Model in statistical mechanics. In this model, in each discrete-time
round, each node looks at a random neighbor and copies her opinion. The Voter
Model is considered the \emph{trivial} opinion dynamics, in the sense that it
is arguably the simplest conceivable rule by which nodes may meaningfully
update their opinion as a function of their neighbors' opinion. Many properties of
this process are understood by known mathematical techniques such as an elegant
duality with the coalescing random walk process~\cite{aldous_reversible_2002}.
In particular, it is known that the Voter Model is not a \emph{fast} dynamics
as for the time it takes before consensus on one opinion is reached in the
network. 
For that reason, the 2-Choices dynamics has been
considered~\cite{cooper_power_2014}. In such dynamics, in each round, each node
looks at the opinion of \emph{two} random neighbors and, if these two are the
same, adopts it. This process can arguably be considered the \emph{simplest} non-trivial type
of opinion dynamics.

The authors of~\cite{cooper_power_2014} consider any initial configuration in
which each node is supporting one out of two possible opinions. They proved that in such a configuration,
under the assumption that the initial bias (i.e., advantage of an opinion) is
greater than a function of the network's \emph{expansion} (measured in terms of
the second eigenvalue of the network)~\cite{hoory_expander_2006}, the whole
network will support the initially most frequent opinion with high probability
after a polylogarithmic number of rounds. 
The results of~\cite{cooper_power_2014} on the 2-Choices dynamics were later
refined with milder assumptions on the initial bias with respect to the network's
expansion~\cite{cooper_fast_2015} and generalized to more opinions~\cite{cooper_fast_2016}. 
Their techniques should be easily adaptable in order to handle
similar dynamics, such as the 3-Majority dynamics~\cite{becchetti_simple_2014,becchetti_stabilizing_2016, ghaffari_tight_2017}. 
The 2-Choices dynamics constitutes one of the few processes whose
behavior has been characterized on non-complete topologies, typically assuming good expansion properties.
Few works analyze the non-trivial behavior of such a dynamics on other classes of topologies which present a clustered structure, e.g., on regular clustered graphs~\cite{CrucianiNS19} and on graphs sampled from the Stochastic Block Models~\cite{shimizu19phase}, proving non-trivial behaviors of the dynamics. In particular, \cite{shimizu19phase} leverages the Kim-Vu concentration of probability inequality and the theory of competitive dynamical systems to obtain a precise description of the phase transition behavior of the 2-Choices dynamics on the aforementioned model.

On a different note, for the deterministic Majority dynamics where in each
round each node updates her opinion with the most frequent one among her
neighbors: Substantial effort has been devoted to investigating how small the
cardinality of a
\emph{monopoly} can be, i.e., a set of nodes supporting a given opinion $\mathcal{O}$ such
that, when running the Majority dynamics, eventually the
network reaches consensus on $\mathcal{O}$~\cite{peleg_size_1998, peleg_local_2002}. 
The previous line of investigation has focused on determining the existence of
network classes on which the monopoly has a size which is upper bounded by a small
function. This question has been settled by~\cite{berger_dynamic_2001}, 
who proved the existence of a family of networks with constant-size monopoly. 
We emphasize that this line of investigation is peculiar as it deals with
existential questions related to specific network classes, as opposed to the
typical research questions that we
discussed so far, which ask for general characterizations of the behavior of
the considered process.

Recently, a more systematic and general study of opinion dynamics has been
carried out in~\cite{mossel_majority_2014, benjamini_convergence_2016,
mossel_opinion_2017}. It characterizes the evolution and
other mathematical aspects of dynamics---such as the Majority dynamics,
the Voter Model, the DeGroot model, and others---on different network classes, such
as Erd\H{o}s-R\'enyi random networks and expander networks. We follow~\cite{mossel_opinion_2017} in adopting the term ``opinion dynamics'' to refer
to the class of processes discussed above. 

A similar perspective to ours has been adopted also in other works.
In~\cite{lambiotte2007majority} the authors show a phase-transition in a mean-field analysis of the behavior of deterministic majority in an asynchronous-update model, in which only few nodes update their opinion at each round, on ``coupled
fully connected networks'' composed of overlapping complete graphs. 
Voter Model and deterministic majority dynamics are also analyzed in~\cite{ijcai2020}, where it is considered a setting in which a superior opinion exists; using tools from the birth-and-death Markov Chains theory, it is shown the interplay between network topology and opinion dynamics in the time needed for the process to reach a consensus on the superior opinion.
In~\cite{icdcn2021} a similar analysis as ours has been carried out for the synchronous $k$-majority dynamics, in which nodes update their opinion to the majority opinion of $k$ randomly selected neighbors, proving similar phase transitions;
in the same work results for the synchronous Voter Model and Majority dynamics are also presented.

\subsection*{Core-Periphery model}
It has long been observed in sociology that many economic and social networks
exhibit a \emph{core-periphery structure}~\cite{snyder_structural_1979,
borgatti_models_2000, rombach_core-periphery_2014,
zhang_identification_2015}, namely, it is generally possible to group nodes into
two classes, a \emph{core} and a \emph{periphery}, such that the former exhibits
a dense
internal topology while the latter is sparse and
loosely connected, with specific properties relating the two.
Such architectural principle has been linked, for example, to the easiness with
which individuals solve routing problems in networks subject to the
small-world phenomenon~\cite{easley_networks_2010}.

The qualitative notion of core-periphery structure was translated into
quantitative relations in the axiomatic approach of%
~\cite{avin_core-periphery_2014}, which later also applied the
algorithmic properties that follow from the core-periphery structure to the
design of efficient distributed networks \cite{avin_distributed_2014}.
In some sense, our theoretical and empirical results may be regarded as a
\emph{functional} justification for the presence of a core-periphery structure
in networks, as the latter turns out to play a decisive role in determining
a certain kind of evolution for basic opinion dynamics such as the 2-Choices.

\section{Analysis of the \advchoiced}
\label{sec:theo}

We now present the \choice{} and an alternative process, based on it, 
that will be analyzed in this section.
The results on such an alternative process will be exploited later, in Section~\ref{sec:core-periphery}, to describe the behavior of the \choice{} on core-periphery networks.
We are going to use use colors, instead of opinions, 
to give an intuitive explanation of the processes 
and facilitate its understanding.

\begin{definition}[\choice]
    \label{def:2choice}
    Given a network $G=(V,E)$ with an initial binary coloring of the \nodes 
    $c: V \mapsto \{red, blue\}$,
    the \choice\ proceeds in synchronous rounds:
    In each round, 
    each \node $u$ chooses two neighbors $v, w$ uniformly at random with replacement; 
    if $c(v)=c(w)$, then $u$ updates its own color to $c(v)$, 
    otherwise $u$ keeps its color.
\end{definition}


\begin{definition}[\advchoice{$p$}{$\sigma$} dynamics]
Let $p \in [0,1]$ be a constant and let $\sigma \in \{red, blue\}$ be a color.
We define the \advchoice{$p$}{$\sigma$} dynamics as a variation of the \choice: 
Every time an \node picks a neighbor, 
with probability $p$ that neighbor supports color $\sigma$ 
regardless of its actual color.
\end{definition}

We start with some notation and preliminaries. 
In the following, we will make use of the asymptotic notation in order to bound the limiting behavior of some of the quantities analyzed in the paper.
In particular we will use $\bigO$, $o$, $\Theta$, $\Omega$, and $\omega$.
This means that our results hold for networks with sufficiently many \nodes.

Let $G=(V,E)$ be a network.
For a set of \nodes $A \subseteq V$, let the \emph{volume} of $A$ be $\vol(A) :=
\sum_{v \in A} d_v$, where $d_v$ is the degree of $v$.
We call $d_{_{\min}} = \min_{v \in V} d_v$.
In this section we only consider dense graphs, i.e., graphs of $n$ \nodes where $d_{_{\min}} = \omega(\log n)$.
Let $\mathbf{c}^{(t)} \in \{red,\, blue\}^n$ 
be the configuration of the colors of the \nodes at time $t$. 
Let the \nodes of $G$ run the \advchoice{p}{blue} dynamics starting from an initial configuration $\mathbf{c}^{(0)}$ where all nodes support color \emph{red}.
Let $B^{(t)}$ be the set of \emph{blue} \nodes
and $R^{(t)} = V \setminus B^{(t)}$ be the set of \emph{red} \nodes at time $t$.
For any \node $v$, let $N_R(v) = N(v) \cap R^{(t)}$ be the set of \emph{red} neighbors and $N_B(v) = N(v) \cap B^{(t)}$ the set of \emph{blue} neighbors of $v$. Furthermore, 
let $r^{(t)}_v$ and $b^{(t)}_v$ respectively be the number of \emph{red} and \emph{blue} neighbors of $v$ at time $t$, i.e., $r^{(t)}_v = |N_R(v)|$ and $b^{(t)}_v = |N_B(v)|$.
Let $\phi^{(t)}_v = \frac{r^{(t)}_v}{d_v}$ be the fraction of \emph{red} \nodes 
in the neighborhood of $v$;
let $\phi_{_{\min}}^{(t)} = \min_{v \in V} \phi^{(t)}_v$
and $\phi_{_{\max}}^{(t)} = \max_{v \in V} \phi^{(t)}_v$ 
be, respectively, the minimum and maximum fractions of \emph{red} neighbors
among all \nodes in $V$.
In the following, for the sake of readability, whenever we omit the time index, 
we refer to the value at time $t$, 
e.g., $\phi_v$ stands for $\phi^{(t)}_v$. 
Similarly, we concisely denote with 
$\PP{R}{v} = \mathbf{P}_{ } ( v \in R^{(t+1)} \,|\, \mathbf{c}^{(t)} = \bar{\mathbf{c}})$
the probability that \node $v$ will be supporting the
\emph{red} color in the next round of the \advchoice{$p$}{\emph{blue}} (conditioned to the event that the configuration at time $t$ is equal to $\bar{\mathbf{c}}$), i.e.,
\[
    \PP{R}{v} = \left\{
    \begin{array}{ll}
    1 - \big( p+(1-p)(1 - \phi_v) \big)^2   & \text{if } v \in R,\\
    (1-p)^2\phi_v^2                         & \text{if } v \in B.
    \end{array}
    \right.
\]
Furthermore, note that:
\begin{align*}
    \min_{w\in R} \PP{R}{w} &= 1 - \big( p+(1-p)(1 - \phi_{\min}) \big)^2,\\
    \min_{w\in B} \PP{R}{w} &= (1-p)^2\phi_{\min}^2.
\end{align*}

Conditioning to the configuration $\mathbf{c}^{(t)} = \bar{\mathbf{c}}$,
we can give a lower bound for the expected fraction of \emph{red} neighbors 
of any \node $v$ as follows:
\begin{align}
    &\Ex{\phi^{(t+1)}_v \cond \mathbf{c}^{(t)} = \bar{\mathbf{c}}} 
    = \frac{1}{d_v} \left(
    \sum_{w \in N_R(v)} \PP{R}{w} + \sum_{w \in N_B(v)} \PP{R}{w} \right) 
    \notag\\
    & \geq 
    \frac{1}{d_v} \left(
    \abs{N_R(v)} \min_{w\in R} \PP{R}{w} 
    + \abs{N_B(v)} \min_{w\in B} \PP{R}{w} \right) 
    \notag\\
    & =
    \frac{r_v}{d_v} \min_{w\in R} \PP{R}{w}  
    + \left(1 - \frac{r_v}{d_v}\right) \min_{w\in B} \PP{R}{w} 
    \notag\\
    & =
    \frac{r_v}{d_v} \left(1 - \big( p+(1-p)(1 - \phi_{_{\min}}) \big)^2 \right)
    + \left(1 - \frac{r_v}{d_v}\right) (1-p)^2\phi_{_{\min}}^2
    \notag\\
    & =
    \frac{r_v}{d_v} \left( 
    1 - \big( p+(1-p)(1 - \phi_{_{\min}}) \big)^2 - (1-p)^2\phi_{_{\min}}^2
    \right)
    + (1-p)^2\phi_{_{\min}}^2
    \notag\\
    & \geq
	\phi_{_{\min}} \left( 
    1 - \big( p+(1-p)(1 - \phi_{_{\min}}) \big)^2 - (1-p)^2\phi_{_{\min}}^2
    \right)
    + (1-p)^2\phi_{_{\min}}^2
    \notag\\
    & =
	\phi_{_{\min}}\left(1 - \big( p+(1-p)(1 - \phi_{_{\min}}) \big)^2 
    + (1-p)^2 (1 - \phi_{_{\min}})\phi_{_{\min}}\right)
    \notag\\
    & =
	\phi_{_{\min}}\left( 1 - 2(1-p)^2 \phi_{_{\min}}^2+(1-p)(3-p) \phi_{_{\min}}-1 \right).
    \notag
\end{align}
Note that we could cancel out $1$ and $-1$, however, leaving them facilitates the analysis. In the steps above, we can lower bound $\frac{r_v}{d_v}$
because its coefficient, i.e.,
$(1 - \big( p+(1-p)(1 - \phi_{_{\min}}) \big)^2 - (1-p)^2\phi_{_{\min}}^2)$,
is non-negative for any $p, \phi_{_{\min}} \in [0,1]$.

Conversely, we can upper bound the expectation using $\phi_{_{\max}}$, i.e.,
\begin{equation}\label{eq:expected_red_fraction}
\Ex{\phi^{(t+1)}_v \cond \mathbf{c}^{(t)} = \bar{\mathbf{c}}} \leq 
\phi_{_{\max}} 
\left( 1 - 2\left(1-p\right)^2 \phi_{_{\max}}^2+(1-p)(3-p) \phi_{_{\max}}-1 \right).
\end{equation}

Note that the lower and the upper bound for the expectation
have the same form. 
In fact, defining the functions
\begin{gather*}
\fp(\phi) := 2\left(1-p\right)^2\phi^2-(1-p)(3-p)\phi+1, \\
\gp(\phi) := \phi ( 1 - \fp(\phi) ), 
\end{gather*}
the lower and the upper bound for the expectation can respectively be written 
as $\gp(\phi_{_{\min}})$ and $\gp(\phi_{_{\max}})$.
Thus, analyzing $\fp(\phi)$, we can see for which values of $p$ 
the function $\gp(\phi)$ is increasing or decreasing.

Before analyzing the actual behavior of the \advchoice{$p$}{\emph{blue}} on $G$, 
we study $\fp(\phi)$ in order to characterize the bounds for the expectation.
The roots of $\fp(\phi)$ are in $\frac{3-p\pm\sqrt{p^2-6p+1}}{4(1-p)}$
while its derivative is
\( 
\fp'(\phi) = 4\left(1-p\right)^2\phi - (1-p)(3-p). 
\)
It follows that $\fp(\phi)$ has a minimum point in 
$\bar{\phi} = \frac{3-p}{4(1-p)}$.
Moreover, the sign of $\fp(\bar{\phi})$ exclusively depends on $p$.
In fact
\begin{align}
\fp(\bar{\phi}) > 0  & \qquad \text{if } p > \D,
\label{eq:positive}\\
\fp(\bar{\phi}) < 0  & \qquad \text{if } p < \D,
\label{eq:negative}
\end{align}
since $\fp(\phi)$ is convex 
and in~\eqref{eq:positive} the discriminant of $\fp(\phi)$ is negative,
while in~\eqref{eq:negative} the discriminant is positive.

\bigskip
Before stating and proving the main theorem of this work, we give the definition of \consensus{} ($\kappa(n)$-consensus), i.e., one of the configurations of interest that are reached by the dynamics.
\begin{definition}[$\kappa(n)$-consensus]
Let $\kappa(n) : \mathbb{N} \rightarrow [0,1]$.
Given a network $G=(V,E)$ and a color $\sigma$, we say a dynamics reaches a $\kappa(n)$-consensus{} in round $t$ if the \emph{volume} of \nodes supporting color $\sigma$ is $(1 - \kappa(n))\vol(V)$.
\end{definition}

We are now ready to state the main theorem, that will be proved in the following subsections.
\begin{theorem}[Phase Transition]\label{thm:noisyd}
Let $G = (V,E)$ be a network of $n$ \nodes such that each \node $v$ 
has a color $\sigma_v$ and $d_v = \omega(\log n)$ neighbors. 
Let $p \in [0,1]$ be a constant and $p^\star = \D$.
Then, starting from a configuration where all \nodes initially support the \emph{red} color
and letting the \nodes run \advchoice{$p$}{\text{blue}},
it holds that:
\begin{itemize}
    \item \emph{$\kappa(n)$-consensus:} If $p > p^\star$, then there exists a constant $\lambda$ such that for every $\kappa(n) : \mathbb{N} \rightarrow [0,1]$ with $\kappa(n) \geq \lambda \cdot \frac{\log n}{d_{_{\min}}}$,
    the dynamics reaches a $\kappa(n)$-consensus on color \emph{blue} within $\bigO(\log(1/\kappa(n)))$ rounds, w.h.p.
    
    \item \emph{Metastability:} If $p < p^\star$, 
    then the \emph{volume} of the \emph{blue} \nodes 
    never exceeds $\frac{1-3p}{4(1-p)} \vol(V)$
    for any $\poly(n)$ number of rounds, w.h.p.
    
    \item \emph{Critical value:} If $p = p^\star$ and $G$ is the complete graph, the dynamics reaches a \emph{consensus} on color \emph{blue} in round $t \in \Omega(\sqrt[4]{n/\log n}) \cap \bigO(n \log n)$, w.h.p.
\end{itemize}
\end{theorem}
Note that in the $\kappa(n)$-consensus result there is a lower bound on the function $\kappa(n) = \Omega(\frac{\log n}{d_{_{\min}}})$. Indeed, given no topological assumption on the graph other than the minimum degree, our proof techniques cannot guarantee a stronger consensus.
As a consequence, the higher the minimum degree, the closer to a complete consensus our analysis can get.


\subsection{\Consensus}
Let $p > \D$. 
Let $\fp(\bar{\phi}) = \epsilon$ be the local minimum of $\fp$. 
Note that $\epsilon$ is positive because of~\eqref{eq:positive} 
and it is a constant since it only depends on $p$ and $\bar{\phi}$, 
which are both constants. 
Due to the convexity of $\fp(\phi)$, it holds that $\fp(\phi) \geq \epsilon$
for every $\phi \in [0,1]$. 
Thus, for every $v \in V$, we have that
\[
  \Ex{\phi^{(t+1)}_v \cond \mathbf{c}^{(t)} = \bar{\mathbf{c}}} 
  \leq \gp(\phi_{_{\max}}^{(t)}) 
    = \phi_{_{\max}}^{(t)} (1 - \fp(\phi_{_{\max}}^{(t)})) 
    \leq \phi_{_{\max}}^{(t)} (1 - \epsilon).
\]
Thus, we can apply a multiplicative form of the Chernoff bounds directly to the
upper bound of the expected fraction of \emph{red} neighbors at the next round, 
as shown in~\cite[Exercise 1.1]{dubhashi_concentration_2009}, and obtain
\begin{align*}
    &\Prob{}{\phi^{(t+1)}_v > (1-\epsilon^2) \phi_{_{\max}}^{(t)} 
        \cond \mathbf{c}^{(t)} = \bar{\mathbf{c}}}
    \\
    &=\Prob{}{\phi^{(t+1)}_v > (1+\epsilon)(1-\epsilon) \phi_{_{\max}}^{(t)} 
        \cond \mathbf{c}^{(t)} = \bar{\mathbf{c}}}
    \\ 
    &= \Prob{}{r^{(t+1)}_v > (1+\epsilon)(1-\epsilon) \phi_{_{\max}}^{(t)} d_v 
        \cond \mathbf{c}^{(t)} = \bar{\mathbf{c}}}
    \\ 
    &\leq e^{-\frac{\epsilon^2}{3}(1 - \epsilon) \phi_{_{\max}}^{(t)} d_v}
    \leq e^{-\frac{\epsilon^2}{3}(1 - \epsilon) \phi_{_{\max}}^{(t)} d_{_{\min}}}
    \\
    &\leq e^{-2\log n} = n^{-2},
\end{align*} 
where the last inequality only holds for configurations
$\mathbf{c}^{(t)} = \bar{\mathbf{c}}$ such that
$\phi_{_{\max}}^{(t)} \geq \kappa(n)$, since by hypothesis we have that $\kappa(n) \geq \lambda \cdot \frac{\log n}{d_{_{\min}}}$ where we set $\lambda = \frac{6}{\varepsilon^2 (1 - \epsilon)}$.
Thus, using the union bound over all the \nodes, we get that 
$\phi_{_{\max}}^{(t+1)} \leq (1-\epsilon^2)\phi_{_{\max}}^{(t)}$, w.h.p.
Formally, for all $t$ such that $\phi_{_{\max}}^{(t)} \geq \kappa(n)$, 
\begin{align*}
\Prob{}{
\exists v \in V: \phi^{(t+1)}_v > (1-\epsilon^2)\phi_{_{\max}}^{(t)} \cond \mathbf{c}^{(t)} = \bar{\mathbf{c}}
} 
&\leq \sum_{ v \in V}\Prob{}{\phi^{(t+1)}_v > (1-\epsilon^2)\phi_{_{\max}}^{(t)} \cond \mathbf{c}^{(t)} = \bar{\mathbf{c}}}
\\
&\leq \sum_{ v \in V} n^{-2} 
= n^{-1}.  
\end{align*}
Starting from a configuration where all \nodes are in state \emph{red} and thus $\phi_{_{\max}}^{(0)} = 1$, we get that after $T$ rounds $\phi_{_{\max}}^{(T)} \leq (1-\varepsilon)^{T}$ and, considering $T=\bigO(\log(1/\kappa(n)))$, it follows that
$\phi_{_{\max}}^{(T)} \leq \kappa(n)$.
Indeed,
\begin{align*}
    \phi_{_{\max}}^{(T)} \leq (1-\varepsilon^2)^T \leq \kappa(n)
    &\iff T \log(1-\varepsilon^2) \leq \log(\kappa(n))
    \\
    &\iff T \geq \log(\kappa(n)) / \log(1-\varepsilon^2)
    \\
    &\iff T \geq \log(1/\kappa(n)) / \log(1/(1-\varepsilon^2)),
\end{align*}
which implies $T = \bigO(\log(1/\kappa(n)))$ since $\varepsilon$ is a positive constant.

Note that if the configuration $\mathbf{c}^{(t)} = \bar{\mathbf{c}}$ is such that $\phi_{_{\max}}^{(t)} < \kappa(n)$, i.e., after $\bigO(\log(1/\kappa(n)))$ rounds of the \advchoice{p}{blue} dynamics, the process has reached a $\kappa(n)$-consensus on color \emph{blue}.
Indeed,
\[
	\vol(R^{(t)}) = \sum_{v \in R} d_v 
	= \sum_{v \in V} (d_v - b_v^{(t)})
	= \sum_{v \in V} d_v \left(1-\frac{b_v^{(t)}}{d_v}\right)
	= \sum_{v \in V} d_v \phi_v^{(t)}
	\leq \phi_{_{\max}}^{(t)} \sum_{v \in V} d_v
	= \kappa(n) \vol(V).
\]
\subsection{Metastability}
Let $p < \D$.
Define $\fp(\bar{\phi}) = -\epsilon$ to be the local minimum of $f$.
Recall that $\epsilon$ is positive because of~\eqref{eq:negative} 
and it is a constant since it only depends on the constants $p$ and $\bar{\phi}$.

Then, using the fact that $\gp(\phi)$ is monotonically nondecreasing 
for $\phi \in [0,1]$,
for every $\phi \geq \bar{\phi}$ we have that
\[ \gp(\phi) \geq \gp(\bar{\phi}) = \bar{\phi}(1- \fp(\bar{\phi})) =  \bar{\phi}(1 + \epsilon). \]

We can now use a multiplicative form of the Chernoff bounds in order to show that
if the fraction of \emph{red} neighbors of an \node $v$ is at least $\bar{\phi}$,
then the probability that the number of \emph{red} neighbors of $v$ in the next round 
is lower than $\bar{\phi}$ is negligible.
Formally, let $\mathbf{c}^{(t)}$ be an arbitrary configuration such that
$\phi_{_{\min}}^{(t)} \geq \bar{\phi}$, 
e.g., the initial configuration $\mathbf{c}^{(0)}$ has this property.
First, note that due to $\gp(\phi) \geq \gp(\bar{\phi})$, we have that
\[
\Ex{\phi^{(t+1)}_v \cond \mathbf{c}^{(t)} = \bar{\mathbf{c}}} \geq \gp(\phi_{_{\min}}^{(t)}) \geq \gp(\bar{\phi}) = \bar{\phi}(1 + \epsilon).
\]
Then, it follows that
\begin{align*}
&\Prob{}{\phi^{(t+1)}_v < \bar{\phi} \cond \mathbf{c}^{(t)}} 
\\
&= \Prob{}{\phi^{(t+1)}_v < \frac{1}{1 + \varepsilon} \bar{\phi} (1 + \varepsilon) \cond \mathbf{c}^{(t)} = \bar{\mathbf{c}}} 
\\
& = \Prob{}{\phi^{(t+1)}_v < \left(1 - \frac{\varepsilon}{1 + \varepsilon}\right) \bar{\phi} (1 + \varepsilon) \cond \mathbf{c}^{(t)} = \bar{\mathbf{c}}} 
\\
& \leq \Prob{}{\phi^{(t+1)}_v < \left(1 - \frac{\varepsilon}{1 + \varepsilon}\right) \Ex{\phi^{(t+1)}_v \cond \mathbf{c}^{(t)} = \bar{\mathbf{c}}} \cond \mathbf{c}^{(t)} = \bar{\mathbf{c}}} 
\\
& = \Prob{}{r^{(t+1)}_v < \left(1 - \frac{\varepsilon}{1 + \varepsilon}\right) d_v \Ex{\phi^{(t+1)}_v \cond \mathbf{c}^{(t)} = \bar{\mathbf{c}}} \cond \mathbf{c}^{(t)} = \bar{\mathbf{c}}} 
\\
& \leq \exp\left(-\frac{\varepsilon^2}{3(1 + \varepsilon)^2} d_v \Ex{\phi^{(t+1)}_v \cond \mathbf{c}^{(t)} = \bar{\mathbf{c}}} \cond \mathbf{c}^{(t)} = \bar{\mathbf{c}} \right) 
\\
& \leq \exp\left(-\frac{\varepsilon^2}{3(1 + \varepsilon)^2} d_v \bar{\phi} (1 + \varepsilon)\right) 
\\
& = e^{-\Omega(d_v)} = e^{-\omega(\log n)} = n^{-\omega(1)}.
\end{align*} 
Applying the union bound over all the \nodes, we get
\begin{align*}
\Prob{}{
	\exists t \in \poly(n),
	\exists v \in V
	: \phi^{(t+1)}_v < \bar{\phi} \cond \mathbf{c}^{(t)} = \bar{\mathbf{c}}
}
&\leq  \sum_{\substack{t \in \poly(n)\\ v \in V}} 
	\Prob{}{\phi^{(t+1)}_v < \bar{\phi} \cond \mathbf{c}^{(t)} = \bar{\mathbf{c}}} 
\\
&= 
\sum_{\substack{t \in \poly(n)\\ v \in V}} n^{-\omega(1)}
= n^{-\omega(1)}.
\end{align*}
Thus, with high probability we have that $\phi_{_{\min}}^{(t)} \geq \bar{\phi}$ for every round polynomial number of rounds. Before we can use this to finish the proof, note that
$\sum_{v \in B^{(t)}} d_v = \sum_{v \in V} (d_v - r_v^{(t)})$, by simply counting the number of \emph{blue} endpoints of an edge in two different ways.
Using that $\phi_v^{(t)} \geq \phi_{_{\min}}^{(t)} \geq \bar{\phi}$ for each $v$, we have
\[
	\vol(B^{(t)}) = \sum_{v \in B} d_v 
	= \sum_{v \in V} (d_v - r_v^{(t)})
		= \sum_{v \in V} d_v \left(1-\frac{r_v^{(t)}}{d_v}\right)
	\leq (1-\bar{\phi}) \sum_{v \in V} d_v
	 = \frac{1-3p}{4(1-p)} \vol(V).
\]
This means, the volume of the \emph{blue} \nodes never exceeds a fraction of
$\frac{1-3p}{4(1-p)}$ of the total volume of the graph, w.h.p.

\subsection{Critical value}

In this section we prove the critical-value part of Theorem \ref{thm:noisyd}.
The analysis of the behavior of the dynamics becomes harder if $p = \D$. Indeed in this case the expected evolution of the dynamics cannot describe accurately its real evolution. 
For this reason we consider the simplest case where the underlying graph $G$ is a clique (assuming it also has self-loops, to further simplify the analysis).
In this scenario we prove both an upper and a lower bound on the converge time of the dynamics to a consensus.%
\footnote{Differently from the previous subsections describing the behavior of the dynamics in the subcritical ($p < p^\star$) and supercritical ($p > p^\star$) regimes, in this setting we are able to talk about full consensus, i.e.\ where all \nodes agree on a same opinion, due to the specific graph topology.}

In the upper bound the main idea is to use the variance of the process to show that, in near-linear time, the color configuration is unbalanced enough towards color \emph{blue} so that the process quickly reaches the blue monochromatic configuration. 
The main part of the lower bound, instead, is where we show that there is a time-frame where the drift toward the blue configuration gets smaller and smaller, and the bound then follows by applying known results on submartigales. 
In both analyses, many events happen with constant probability, so it is not possible to apply the same technique used in the previous sections where $p \neq p^\star$, i.e.\ giving bounds neighbourhood by neighbourhood and ``merge'' them through the Union Bound. 

Even if we restrict the analysis to the case where $G$ is the complete graph, we observe that in this context the convergence time of the dynamics is quite different from both the previous cases, being closer to the convergence time of Voter~\cite{becchetti2020overview}.


\paragraph*{Upper bound.}
We first show an upper bound to the time that the process needs to reach a configuration in which the number of blue nodes is big enough, e.g., 3/4 of the total, exploiting the variance of the process.
Then, starting from a configuration where 3/4 of the nodes are blue, we can rely on previous results on the unbiased version of the \choice{}~\cite{cooper_fast_2015}, which is stochastically dominated by the \advchoice{p}{blue} (having no help of the bias $p$ in reaching a blue consensus and, thus, being slower).
\begin{enumerate}
\item Let $B^{(t)} = n (1 - \phi^{(t)})$ be the random variable counting the number of blue nodes at time $t$ and let $\tau := \inf\{t \in \mathbb{N} : B^{(t)} \geq \frac{3}{4}n\}$ be a stopping time. 
We first show in Lemma~\ref{lemma:a botte de varianza} that starting from any configuration where $B^{(0)} < \frac{3}{4}n$, it holds that $\Ex{\tau} = \bigO(n)$.
As a simple application of the Markov inequality we get that $\Prob{}{\tau > 2\Ex{\tau}} \leq \frac{\Ex{\tau}}{2\Ex{\tau}} \leq \frac{1}{2}$, using Lemma~\ref{lemma:a botte de varianza}. As an application of the Union Bound we get that, w.h.p., the hitting condition is satisfied at least once within $\log n$ attempts. This shows that the stopping time $\tau = \bigO(n \log n)$ w.h.p., namely the process reaches a configuration such that the number of blue nodes is at least $\frac{3}{4}n$ within $\bigO(n \log n)$, w.h.p.

\item Starting from such a set of configurations we apply the results on the \choice{}~\cite[Theorem 1]{cooper_fast_2015}, where the authors show that the dynamics converges to the all-blue configuration in $\bigO(\log n)$ rounds, w.h.p., if the initial unbalance toward the blue color is big enough with respect to some function of the spectral expansion of the graph.
\end{enumerate}

\bigskip
In detail, we first apply our previous equations on the complete graph $G$. 
Since $G$ is the complete graph with self-loops we have that the set of neighbors for each
$v \in V$ is $N_v=V$. Thus, for each $v,w \in V$ the \emph{number} of red neighbors is equal, i.e. $r_v = r_w$, and therefore also the \emph{ratio} of red neighbors is equal, i.e. $\phi_v = \phi_w$. 
For simplicity, we will just look at the number of red neighbors $r^{(t)}$
and at the fraction of red neighbors 
\[
    \phi^{(t)} := \frac{r^{(t)}}{n}.
\]
The topology of $G$ also implies that a configuration 
$\mathbf{c}^{(t)}$ at time $t$ is completely determined by the number (or fraction) of red nodes at time $t$.
Thus, for this scenario, it follows from Equation~\eqref{eq:expected_red_fraction} that
\begin{align}\label{eq:expected_red_frac_clique}
\Ex{\phi^{(t+1)} \cond \phi^{(t)}=\phi} = \phi\left(1-f_{p^\star}(\phi)\right),
\end{align}
since the minimum and maximum fraction of red neighbors are the same for all nodes and equal among them, i.e. $\phi_{_{\min}} = \phi_{_{\max}} = \phi$.
Since in this scenario 
\[
    p = p^\star = \D{},
\] 
we get that
\begin{align}
    f_{p^\star}(\phi) 
    &= 2\left(1-p^\star\right)^2\phi^2-(1-p^\star)(3-p^\star)\phi+1
    \notag\\
    &= (2\sqrt{2}(\sqrt{2}-1)\phi -1)^2.
    \label{eq:fpstar}
\end{align}
Note that, since $f_{p^\star}$ is a square, it holds that
$\Ex{\phi^{(t+1)} \cond \phi^{(t)}=\phi} \leq \phi^{(t)}$.
Thus the converse must also hold, i.e.\ 
\begin{equation}\label{eq:blue_fraction_submartingale}
    \Ex{1-\phi^{(t+1)} \cond \phi^{(t)}=\phi} 
    \geq 1-\phi^{(t)}.
\end{equation}

\begin{lemma}\label{lemma:a botte de varianza}
    Let $B^{(t)}$ the random variable counting the number of blue nodes at time $t$ and let $\tau := \inf\{t \in \mathbb{N} : B^{(t)} \geq \frac{3}{4}n\}$ be a stopping time. 
    Starting from any configuration where $R^{(0)} \geq \frac{1}{4}n$, it holds that $\Ex{\tau} = \bigO(n)$.
\end{lemma}
\begin{proof}
    We start by proving that 
    \begin{equation}\label{eq:lower_bound_cond_variance}
        \Var{B^{(t+1)} \cond \phi^{(t)}=\phi} 
        \geq R^{(t)} \xi_r (1-\xi_r)
        \geq \delta n,
    \end{equation}
    where $\delta := \frac{(p^\star)^2(1-p^\star)}{16}$ is constant w.r.t.\ $n$.
    The core of the proof are indeed Equations~\eqref{eq:blue_fraction_submartingale} and~\eqref{eq:lower_bound_cond_variance}, 
    that allow us to show that when the number of blue nodes is sufficiently large, i.e.\  $B \geq \frac{3}{4}n$, the dynamics is a submartingale with conditional variance $\Theta(n)$. 
    Then we show that $\tau := \inf\{t \in \mathbb{N} : B^{(t)} \geq \frac{3}{4}n\}$ 
    is a hitting time with linear expectation w.r.t.\ the number of nodes $n$.
    
    Since $B^{(t)} := n(1-\phi^{(t)})$, Equation~\eqref{eq:blue_fraction_submartingale} implies that
    \[
        \Ex{B^{(t+1)} \cond \phi^{(t)}=\phi} 
        \geq B^{(t)}.
    \]
    Let $X(v)$ be an indicator random variable for the event 
    \enquote{node $v$ is \emph{blue} in the next round}.
    We have $X(v) = 1$ if $v$ is currently red and sees two blue nodes; we denote the probability of this event with $\xi_r$.
    We also have $X(v) = 1$ if $v$ is currently blue and does not see two red nodes; we denote the probability of this event with $\xi_b$.
    By indicating with $\text{Bin}(n,\,p)$ a random variable sampled from a binomial distribution with parameters $n$ (trials) and $p$ (probability of success), 
    it follows that
    \begin{align*}
        \Var{B^{(t+1)} \cond \phi^{(t)}=\phi}
        &= \Var{\sum_{v \in V} X(v) \cond \phi^{(t)}=\phi}
        \\
        &= \Var{\sum_{v \in R^{(t)}} X(v) 
            + \sum_{v \in B^{(t)}} X(v) \cond \phi^{(t)}=\phi}
        \\
        &= \Var{\text{Bin}(R^{(t)},\,\xi_r) + \text{Bin}(B^{(t)},\,\xi_b) \cond \phi^{(t)}=\phi}
        \\
        &= \Var{\text{Bin}(R^{(t)},\,\xi_r) \cond \phi^{(t)}=\phi}
            + \Var{\text{Bin}(B^{(t)},\,\xi_b) \cond \phi^{(t)}=\phi}
        \\
        &\geq \Var{\text{Bin}(R^{(t)},\,\xi_r) \cond \phi^{(t)}=\phi}
        \\
        &= R^{(t)} \xi_r (1-\xi_r).
    \end{align*}
    Recall that, by hypothesis, $R^{(t)} \geq \frac{1}{4}n$.
    Moreover, note that:
    i) $\xi_r$ (i.e.\  the probability of seeing blue twice), 
    is greater than that of seeing blue twice only considering the bias
    and thus $\xi_r \geq (p^\star)^2$;
    ii) $(1-\xi_r)$, (i.e.\  the probability of seeing at least one red), 
    is greater than the probability that the first node that we see is red,
    thus
    \(
        (1-\xi_r) \geq (1-p^\star)\phi \geq \frac{1-p^\star}{4}.
    \)
    Therefore it follows that
    \(
        \Var{B^{(t+1)} \cond \phi^{(t)}=\phi} 
        \geq R^{(t)} \xi_r (1-\xi_r)
        \geq \frac{(p^\star)^2(1-p^\star)n}{16}
        = \delta n,
    \)
    where $\delta := \frac{(p^\star)^2(1-p^\star)}{16}$.
    
    \bigskip
    Let us now define a new random variable $Z^{(t)} := (B^{(t)})^2 - \delta n t$.
    We show that the random process defined by $\{Z^{(t)}\}_{t \in \mathbb{N}}$ is 
    a submartingale. Observe that
    \begin{align*}
        \Ex{Z^{(t+1)} \cond \phi^{(t)}=\phi} 
        &= \Ex{(B^{(t+1)})^2 \cond \phi^{(t)}=\phi} - \delta n (t+1)
        \\
        &= \Ex{(B^{(t+1)}) \cond \phi^{(t)}=\phi}^2 + \Var{(B^{(t+1)}) \cond \phi^{(t)}=\phi} - \delta n (t+1)
        \\
        &\stackrel{(a)}{\geq} \Ex{(B^{(t+1)}) \cond \phi^{(t)}=\phi}^2 + \delta n - \delta n (t+1)
        \\
        &= \Ex{(B^{(t+1)}) \cond \phi^{(t)}=\phi}^2 - \delta n t
        \\
        &\stackrel{(b)}{\geq} (B^{(t)})^2 - \delta n t
        = Z^{(t)},
    \end{align*}
    where in $(a)$ we used the lower bound on the conditional variance
    (Equation~\eqref{eq:lower_bound_cond_variance}),
    and in $(b)$ we used the fact that the fraction of blue nodes is a
    submartingale (Equation~\eqref{eq:blue_fraction_submartingale}).
    
    Let us define the stopping time 
    $\tau := \inf\{t \in \mathbb{N} : B^{(t)} \geq \frac{3}{4}n\}$.
    Note that $\tau$ meets the conditions of the 
    Optional Stopping Theorem~\cite[Corollary~17.7]{levin2017markov},
    in its version for submartingales.
    Indeed:
    \begin{itemize}
        \item $\Prob{}{\tau < \infty} = 1$ holds because the probability of jumping
        in the configuration where all nodes are blue at any time time $t$ is lower
        bounded by a nonzero value independent from $t$, namely
        $\Prob{}{b^{(t+1)}=n} \geq (p^\star)^{2n}$ (in words, 
        it is at least the probability for each node to see twice color blue).
    
        \item $|Z^{(t)}| \leq K$, for every $t \leq \tau$ 
        and for some value $K$ independent from $t$, holds because 
        $|Z^{(t)}| = |(B^{(t)})^2 - \delta n t| \leq |(B^{(t)})^2| + |\delta n t| \leq n^2 + cn$, where the last inequality holds from previous point, i.e.\  since $\tau < \infty$ almost surely.
    \end{itemize}
    The Optional Stopping Theorem states that $\Ex{Z^{(\tau)}} \geq \Ex{Z^{(0)}}$;
    it then follows that
    \[
        \Ex{Z^{(\tau)}} = \Ex{(B^{(\tau)})^2} - \delta n \Ex{\tau}
        \geq \Ex{Z^{(0)}} = 0
    \]
    and, thus, we can conclude that
    \[
        \Ex{\tau} \leq \frac{\Ex{(B^{(\tau)})^2} - \Ex{Z^{(0)}}}{\delta n}
        = \frac{\Ex{(B^{(\tau)})^2}}{\delta n}
        \leq \frac{n^2}{\delta n} 
        \leq \frac{n}{\delta}.
    \]
\end{proof}


We now apply previous results on the \choice{} dynamics~\cite[Theorem 1]{cooper_fast_2015}, that claim that the \choice{} dynamics converges in $\bigO(\log n)$ rounds to a consensus on the initial majority, w.h.p., if the second largest eigenvalue of the transition matrix of a random walk on the underlying graph is less or equal to $1/\sqrt{2}$, and if the difference between the number of supporters of blue and red colors is $\Omega(n)$.
Thanks to Lemma~\ref{lemma:a botte de varianza}, we know that after $\bigO(n \log n)$ rounds the dynamics reaches a configuration where $|B^{(t)}| \geq \frac{3}{4} n$, even if starting from an initial worst-case configuration where all nodes supported state red. This implies the $\Omega(n)$ condition on difference between blue and red colors support.
Moreover note that the transition matrix of the clique with self loops is $P = \frac{1}{n} J$, where $J$ is the $n \times n$ matrix of all ones and the eigenvalues of $P$ are $1$, with multiplicity 1, and $0$ (including the second largest), with multiplicity $n-1$.

Since the \choice{} dynamics is equivalent to the \advchoice{p}{blue} when $p=0$, it follows from a stochastic domination argument that also the \advchoice{p}{blue} dynamics converges in $\bigO(\log n)$ rounds to a blue consensus, w.h.p., since we lowered the bias $p$ from $p=\D$ to $p=0$.


\paragraph*{Lower bound.}
The proof for the lower bound is divided in two parts.
Note from Equation~\eqref{eq:fpstar} that $f_{p^\star}$ has a root in $\bar{\phi} = \frac{1}{2\sqrt{2}(\sqrt{2}-1)}$. 
Combining this with Equation~\eqref{eq:expected_red_frac_clique}, it follows that when $\phi^{(t)} = \bar{\phi}$ the process is stationary in expectation, namely $\Ex{\phi^{(t+1)} \cond \phi^{(t)}=\bar{\phi}} = \bar{\phi}$.
We essentially show a lower bound on the time the process needs to reach a configuration where the number of red nodes is $\bar{\phi}n$, starting from the initial configuration in which all nodes are red.
In order to do this we use the following strategy: 
\begin{enumerate}
\item We look at the evolution over time of the difference between the number of red nodes and the target number of red nodes, namely 
\begin{equation}\label{eq:deltadefinition}
    \Delta^{(t)} := R^{(t)} - \bar{\phi} n.
\end{equation}
Note that the closer $\Delta^{(t)}$ is to zero, the smaller is the drift of the dynamics since it goes closer to the saddle point $\bar{\phi}$ where it is a martingale.

\item In Lemma~\ref{lemma:slowly to the saddle} we prove an upper bound to the difference $\Delta^{(t)}-\Delta^{(t+1)}$ in two consecutive rounds, namely that 
$\Delta^{(t)}-\Delta^{(t+1)}\leq [8p^* + \sqrt{2}]\sqrt{n\log n}$ (we recall that $p^* = 3 - 2\sqrt{2}$), starting from a configuration in which $\Delta^{(t)} \leq n^{\frac{3}{4}}\sqrt[4]{\log n}$.
The lemma allows us to give a lower bound on the time that the process needs to reach the saddle point in which $\Delta^{(t)}$ is close to $0$, w.h.p. 
In fact, if the process lies in a configuration such that $\frac{1}{2} n^{\frac{3}{4}}\sqrt[4]{\log n} \leq \Delta^{(t)} \leq n^{\frac{3}{4}}\sqrt[4]{\log n}$, 
it needs $({\frac{1}{2} n^\frac{3}{4}\sqrt[4]{\log n}})/({[8p^* + \sqrt{2}]\sqrt{n\log n}}) = \Omega(\sqrt[4]{n/\log n})$ rounds to reach a configuration such that $\Delta^{(t)} = \bigO(1)$, w.h.p.

\item In Lemma~\ref{lemma:small jumps}, instead, we prove that at any time $t$ such that $\Delta^{(t)} = \omega(n^{\frac{3}{4}}\sqrt[4]{\log n})$, e.g., at the beginning of the process in which $\Delta^{(0)} = \Theta(n)$, then the process either jumps in the desired range of values for $\Delta$ or the dynamics does not converge at all, w.h.p.
\end{enumerate}
The proof follows from the two lemmas and from an application of the Union Bound for $\tilde\bigO(n^{\frac{1}{4}})$ rounds.

\begin{lemma}\label{lemma:slowly to the saddle}
If $\Delta^{(t)} \leq n^\frac{3}{4}\sqrt[4]{\log n}$, then $(\Delta^{(t)}-\Delta^{(t+1)}) \leq [8p^* + \sqrt{2}]\sqrt{n\log n}$.
\end{lemma}
\begin{proof}
We first lower bound the decrease of the number of red nodes at each round by using a Chernoff Bound (Equation~\ref{eq:cb_red}); then we express the result as a function of $\Delta^{(t)}$. 

Let $R^{(t)}$ be the random variable counting the number of red nodes at time $t$.
Let $\epsilon := \sqrt{\frac{2\log n}{\Ex{R^{(t+1)}}}}$.
Note that $\varepsilon \in (0,1)$ since $\Delta^{(t)} \leq n^\frac{3}{4}\sqrt[4]{\log n}$ that implies that $\Ex{R^{(t+1)}} > 2 \log n$.

By applying a multiplicative form of the Chernoff bound, it holds that
\begin{equation}\label{eq:cb_red}
    \Prob{}{R^{(t+1)} < \Ex{R^{(t+1)}} (1-\epsilon) \cond \Delta^{(t)} \leq n^\frac{3}{4}\sqrt[4]{\log n}}
    \leq e^{-\frac{\epsilon^2}{2}\Ex{R^{(t+1)}}}
    = \frac{1}{n}.
\end{equation}

Recall that $\bar{\phi} = \frac{1}{2\sqrt{2}(\sqrt{2}-1)}$ is the root of $f_{p^\star}(\phi)$.
Let us define the extra number of red nodes w.r.t.\ the critical number, i.e.\  the number of red nodes such that $f_p(\phi) = 0$, as

Thanks to Equation~\ref{eq:cb_red} the following holds w.h.p.
\begin{align}
    && R^{(t+1)} &\geq \Ex{R^{(t+1)}}(1 - \epsilon)
    \notag\\
    \iff && R^{(t+1)} - \bar{\phi}n &\geq \Ex{R^{(t+1)}}(1 - \epsilon) - \bar{\phi}n
    \notag\\
    \iff && \Delta^{(t+1)} &\geq \Ex{\Delta^{(t+1)}} - \epsilon \Ex{R^{(t+1)}}
    \label{eq:deltat+1}\\
    \iff && -\Delta^{(t+1)} &\leq - \Ex{\Delta^{(t+1)}} + \epsilon \Ex{R^{(t+1)}}
    \notag\\
    \iff && \Delta^{(t)}-\Delta^{(t+1)} &\leq \Delta^{(t)} - \Ex{\Delta^{(t+1)}} + \epsilon \Ex{R^{(t+1)}}
    \notag\\
    \iff && \Delta^{(t)}-\Delta^{(t+1)} &\leq \Delta^{(t)} - \Ex{\Delta^{(t+1)}} + \sqrt{2\log n} \sqrt{\Ex{R^{(t+1)}}}
    \notag\\
    \iff && \Delta^{(t)}-\Delta^{(t+1)} &\stackrel{(a)}{\leq} \Delta^{(t)} - \Ex{\Delta^{(t+1)}} + \sqrt{2n\log n},
    \notag
\end{align}
where in $(a)$ we use that $\Ex{R^{(t+1)}} \leq n$.
In order to conclude the proof, we show in the next claim a suitable formulation for the expectation of $\Delta^{(t+1)}$ that will be used in the last steps.
\begin{quote}
    \begin{claim}\label{claim:exp_delta}
    It holds that
    \begin{equation*}
       \Ex{\Delta^{(t+1)} \cond \Delta^{(t)}, R^{(t)}} = \Delta^{(t)}\left(1- 8p^*\frac{\Delta^{(t)}}{n}\frac{R^{(t)}}{n}\right).
    \end{equation*}
    \end{claim}
    \begin{proof}
    By using Equation~\eqref{eq:expected_red_frac_clique}, we get that
    \begin{align*}
        \Ex{\Delta^{(t+1)}} &= \Ex{R^{(t+1)}} - \bar{\phi}n
        \\
        &= R^{(t)}\left[1 - f_{p^\star}\left(\frac{R^{(t)}}{n}\right) \right] - \bar{\phi}n
        \\
        &= R^{(t)} - R^{(t)}\cdot f_{p^\star}\left(\frac{R^{(t)}}{n}\right) - \bar{\phi}n
        \\
        &= \Delta^{(t)} - R^{(t)}\cdot f_{p^\star}\left(\frac{R^{(t)}}{n}\right)
        \\
        &= \Delta^{(t)} \left[ 1 - \frac{R^{(t)}}{\Delta^{(t)}}\cdot f_{p^\star}\left(\frac{R^{(t)}}{n}\right)\right].
    \end{align*}
    Moreover, by using the definition of $\Delta$ (Equation~\ref{eq:deltadefinition}), we have that
    \begin{align*}
        f_{p^\star}\left(\frac{R^{(t)}}{n}\right) &= \left( 2\sqrt{2}(\sqrt{2}-1) \frac{R^{(t)}}{n} - 1 \right)^2
        \\
        &= \left( 2\sqrt{2}(\sqrt{2}-1) \left( \frac{\Delta^{(t)} + \bar{\phi}n}{n} \right) - 1 \right)^2
        \\
        &= \left( 2\sqrt{2}(\sqrt{2}-1) \left( \frac{\Delta^{(t)}}{n} + \bar{\phi} \right) - 1 \right)^2
        \\
        &= \left( 2\sqrt{2}(\sqrt{2}-1) \frac{\Delta^{(t)}}{n} + 1 - 1 \right)^2
        \\
        &= 8(\sqrt{2}-1)^2 \left(\frac{\Delta^{(t)}}{n}\right)^2.
    \end{align*}
    
    By combining the last two equations and considering that $(\sqrt{2}-1)^2 = 3 - 2\sqrt{2} = p^\star$ we get that
    \begin{align}
        \Ex{\Delta^{(t+1)}} &= \Delta^{(t)} \left[ 1 - \frac{R^{(t)}}{\Delta^{(t)}}\cdot 8(\sqrt{2}-1)^2 \left(\frac{\Delta^{(t)}}{n}\right)^2\right]
        \notag\\
        \label{eq:expdelta}
        &= \Delta^{(t)} \left[ 1 - 8p^* \frac{\Delta^{(t)}}{n} \frac{R^{(t)}}{n} \right].
    \end{align}
    \end{proof}
\end{quote}

\noindent
Using Claim \ref{claim:exp_delta} we we conclude that
\begin{align*}
    \Delta^{(t)}-\Delta^{(t+1)} 
    &\leq \Delta^{(t)} - \Delta^{(t)} + 8p^*\frac{(\Delta^{(t)})^2}{n}\frac{R^{(t)}}{n} + \sqrt{2n\log n}
    \\
    &\stackrel{(a)}{\leq} 8p^*\frac{(\Delta^{(t)})^2}{n} + \sqrt{2n\log n}
    \\
    &\stackrel{(b)}{\leq} [8p^* + \sqrt{2}]\sqrt{n\log n},
\end{align*}
where in $(a)$ we use that $R^{(t)} \leq n$ and in $(b)$ that $\Delta^{(t)} \leq n^\frac{3}{4}\sqrt[4]{\log n}$.
\end{proof}

\begin{lemma}\label{lemma:small jumps}
The process does not jump from a configuration where $\Delta \geq n^\frac{3}{4}\sqrt[4]{\log n}$ to a configuration where $\Delta < \frac{1}{2} n^\frac{3}{4}\sqrt[4]{\log n}$, w.h.p.
\end{lemma}
\begin{proof}
Starting from Equation~\eqref{eq:expdelta} we get
\begin{align}
    \Ex{\Delta^{(t+1)}} 
    &= \Delta^{(t)} \left[ 1 - 8(\sqrt{2}-1)^2 \frac{\Delta^{(t)}}{n} \frac{R^{(t)}}{n} \right]
    \notag\\
    &= \Delta^{(t)} \left[ 1 - 8(\sqrt{2}-1)^2 \cdot \left(\frac{R^{(t)}-\bar{\phi}n}{n}\right) \frac{R^{(t)}}{n} \right]
    \notag\\
    &\geq \Delta^{(t)} \left[ 1 - 8(\sqrt{2}-1)^2 \cdot \left(\frac{n-\bar{\phi}n}{n}\right) \right]
    \notag\\
    &\geq \Delta^{(t)} \left[ 1 - 8(\sqrt{2}-1)^2 \cdot (1-\bar{\phi}) \right].
    \label{eq:lb_deltat+1}
\end{align}
Note that $1 - 8(\sqrt{2}-1)^2 \cdot (1-\bar{\phi}) \approx 0.8 > 0$.
Thus, conditioning with respect to the complement of the event of Equation~\eqref{eq:cb_red} and using 
Equation~\eqref{eq:deltat+1} as starting point, we get the following bound to the value of the random variable $\Delta^{(t+1)}$, assuming a suitably large value of $n$:
\begin{align*}
    \Delta^{(t+1)} &\geq \Ex{\Delta^{(t+1)}} - \epsilon \Ex{R^{(t+1)}}
    \\
    &\stackrel{(a)}{\geq} \Ex{\Delta^{(t+1)}} - \sqrt{2n\log n}
    \\
    &\stackrel{(b)}{\geq} (1 - 8(\sqrt{2}-1)^2 \cdot (1-\bar{\phi})) \Delta^{(t)} - \sqrt{2n\log n}
    \\
    &\stackrel{(c)}{\geq} (1 - 8(\sqrt{2}-1)^2 \cdot (1-\bar{\phi}) - \epsilon') \Delta^{(t)}
    \\
    &\geq \frac{1}{2}\Delta^{(t)}
    \\
    &\geq \frac{1}{2}n^\frac{3}{4}\sqrt[4]{\log n},
\end{align*}
where in $(a)$ we used that $\epsilon = \sqrt{\frac{2\log n}{\Ex{R^{(t+1)}}}}$ and that $\Ex{R^{(t+1)}} \leq n$, in $(b)$ we used Inequality~\eqref{eq:lb_deltat+1},
and in $(c)$ we used that $\sqrt{n\log n} = o(\Delta^{(t)})$ since $\Delta^{(t)} \geq n^\frac{3}{4}\sqrt[4]{\log n}$ and thus for any $\epsilon' > 0$ there exists a suitably large $n$.
\end{proof}

\section{Results on Core-Periphery Networks}
\label{sec:core-periphery}
In this section we present theoretical and experimental results about the behavior of the \choice{} on networks exhibiting a core-periphery structure.
In particular, in Section \ref{sec:theo-cp} we apply the technical result of Theorem~\ref{thm:noisyd} on a suitably defined class of Core-Periphery networks.
In Section \ref{sec:experiments}, instead, we introduce a new heuristic to identify a good core-periphery partition and simulate the behavior of the \choice{} in the same setting we theoretically analyzed.


\subsection{Theoretical Results}
\label{sec:theo-cp}

Let us start by giving a formal definition of Core-Periphery networks on which it is possible to apply the results of Theorem~\ref{thm:noisyd}.
\begin{definition}
    \label{def:regcore}
    Let $n \in \mathbb{N}$ and $\epsilon,c_r,c_d \in \mathbb{R}^+$, with
    $ \frac{1}{2} \leq \epsilon \leq 1$.
	We define an $(n, \epsilon, c_r,c_d)$-Core-Periphery network $G=(V, E)$
    as a network where each \node{} either belongs to the core $\core{}$ 
    or to the periphery $\periphery{}$, i.e., $V = \core \dijcup \periphery$,
    with $\abs{\core} = n^\epsilon$ and $\abs{\periphery} = n$,
    and such that:
    \begin{itemize}
        \item for each \node $u \in \core$, it holds that
        $\vert N(u) \cap \core\vert = c_r \cdot \vert N(u) \cap \periphery\vert$,
        \item for each \node $v \in \periphery$, it holds that
        $\vert N(v) \cap \core\vert = c_d \cdot \vert N(v) \cap \periphery\vert$,
    \end{itemize}
    where $N(v)$ is the set of neighbors of \node $v$.
\end{definition}

Note that the definition we just provided matches the requirements
of the core-periphery structure as axiomatized by Avin et
al.~\cite{avin_core-periphery_2014}.
However, observe that it is more restrictive: The values $c_r$ and $c_d$ define
properties that hold for each \node of the network and not only globally, i.e.,
for the partition induced by the core.

Let $G = (V, E)$ be an $(n, \epsilon, c_r, c_d)$-Core-Periphery network. 
We consider the natural setting in which, w.l.o.g., 
the \nodes belonging to the core $\core$ initially support the color \emph{blue}
while the remaining \nodes, from the periphery $\periphery$, support the color \emph{red}.
With the following theorem, we prove a phase transition of the \choice\ on Core-Periphery networks, showing how to apply the technical results provided in Theorem \ref{thm:noisyd}.

\begin{theorem}
    \label{thm:phase_transition_cp}
    Let $c^\star = \frac{\sqrt{2}-1}{2}$ be a universal constant.
    Let $G=(V, E)$ be an $(n, \epsilon, c_r, c_d)$-Core-Periphery network
    such that $d_{_{\min}} = \omega(\log n)$ and let $G$ run the \choice{}.
    It holds that:
    \begin{itemize}
        \item $\kappa(n)$-consensus: There exists a constant $\lambda$ such that for all $\kappa(n) : \mathbb{N} \rightarrow [0,1]$ with $\kappa(n) \geq \lambda \cdot \frac{\log n}{d_{_{\min}}}$, there exists a round $t \in \bigO(\log(1/\kappa(n)))$ such that:
        \begin{itemize}
            \item if $c_r > n^{(\epsilon + \delta)/2}$ and $c_d > c^\star$ by a constant, then $\vol(B^{(t)}) = (1 - \kappa(n))\vol(V)$, w.h.p.
        \end{itemize}
        \item Metastability: 
        For each round $t \in \poly(n)$ it holds that:
        \begin{itemize}
            \item if $c_r > \frac{1}{c^\star}$ by a constant, then $\vol(B^{(t)}) \geq \frac{3}{4}\vol(\core)$, w.h.p.
            \item if $c_d < c^\star$ by a constant, then $\vol(R^{(t)}) \geq \frac{3}{4}\vol(\periphery)$, w.h.p.
        \end{itemize}
    \end{itemize}
\end{theorem}

\begin{proof}
Let $q_{_\core} = \frac{| N(u) \cap \periphery |}{d_u}$
be the probability that an agent $u \in \core$ picks a neighbor in the periphery,
and let $q_{_\periphery} = \frac{| N(v) \cap \core |}{d_v}$
be the probability that an agent $v \in \periphery$ picks a neighbor in the core.
The relations below follow from Definition~\ref{def:regcore}:
\begin{align}
c_r &= | N(u) \cap \core | \,/\, | N(u) \cap \periphery |
    = (1-q_{_\core}) \,/\, q_{_\core}
& \forall u \in \core, 
\label{eq:robustnessprobability} 
\\
c_d &= | N(v) \cap \core | \,/\, | N(v) \cap \periphery |
    = q_{_\periphery} \,/\, (1-q_{_\periphery})
& \forall v \in \periphery.
\label{eq:dominanceprobability}
\end{align}

Let $c^\star = \frac{\sqrt{2} - 1}{2}$ be the constant which later defines a threshold
between \emph{\consensus{}} and \emph{metastability} behavior.
We get
\begin{align}
	c_r = 1 \,/\, (c^\star + \delta_r) \quad&\implies\quad q_{_\core} = \D + \delta_r' \label{eq:robustnessprobability_rel} \\
	c_d = c^\star + \delta_d \quad&\implies\quad q_{_\periphery} = \D + \delta_d' \label{eq:dominanceprobability_rel}
\end{align}
for $\delta_r$ and $\delta_r'$ ($\delta_d$ and $\delta_d'$) which are either both positive or both negative.

\paragraph{\Consensus.}
For the \consensus\ result, we require a high robustness of the core such that it
remains monochromatic for $\bigO(\log(1/\kappa(n)))$ rounds, where $\kappa(n) : \mathbb{N} \rightarrow [0,1]$ with $\kappa(n) \geq \lambda \cdot \frac{\log n}{d_{_{\min}}}$ by hypothesis. 
The following lemma is needed to link the robustness with this property.
\begin{quote}
    \begin{lemma}
        Let $\epsilon$ and $\delta$ be two positive constants. 
        Let $G = (V,E)$ be a graph of $n^\epsilon$ agents, 
        and let $0 \leq p \leq n^{-{(\epsilon + \delta)}/{2}}$. 
        Starting from a configuration such that each \node\ 
        initially supports the \emph{blue} color, 
    	within $\bigO(\log(1/\kappa(n)))$ rounds of the \advchoice{$p$}{red}
        no agent becomes \emph{red}, w.h.p.
        \label{lemma:core}
    \end{lemma}
    \begin{proof}
    	The probability that an agent $v$ changes its color to \emph{red} at time $t$, 
    	given that all the other agents are still \emph{blue}, is
        \begin{align*}
        & \Prob{}{v \in R^{(t+1)} \cond V = B^{(t)}} 
        = p^2
        \leq \left(n^{-(\epsilon + \delta)/2}\right)^2 
        = n^{-(\epsilon + \delta)}.
        \end{align*}
    	Applying the union bound over all the agents and over $\tau = \mathcal{O}(\log(1/\kappa(n)))$ rounds, we get
        \begin{align*}
        &\Prob{}{
            \exists t \leq \tau,\exists v \in V 
            : v \in R^{(t+1)} \cond V = B^{(t)}} 
            \leq \frac{n^\epsilon \cdot \tau}{n^{\epsilon+\delta}} 
            = \bigO(n^{-\delta/2}),
        \end{align*}
        since $\tau = \mathcal{O}(\log(1/\kappa(n)))$
        and $\kappa(n) \geq \lambda \cdot \frac{\log n}{d_{_{\min}}}$.
        Thus, all agents in the graph remain \emph{blue} 
    	for any logarithmic number of rounds, w.h.p.
    \end{proof}
\end{quote}
If $c_r \geq n^{(\epsilon + \delta)/2}$, by Equation~\eqref{eq:robustnessprobability} it follows
that
\(
	q_{_\core} = 1 \,/\, (c_r + 1) < 1 \,/\, c_r \leq n^{-(\epsilon + \delta)/2}.
\)
Thus, we can apply Lemma~\ref{lemma:core} 
obtaining that the agents in the core never changes color for $\mathcal{O}(\log(1/\kappa(n)))$
rounds, w.h.p.
Therefore, for any $\mathcal{O}(\log(1/\kappa(n)))$ number of rounds, the process is equivalent to a \advchoice{$q_{_\periphery}$}{$blue$}
run by the periphery. 
Since $c_d > c^\star$ and thus $q_{_\periphery} > \D$ by Equation \eqref{eq:dominanceprobability_rel},
we can apply Theorem~\ref{thm:noisyd} and get a $\kappa(n)$-consensus on the \emph{blue} color within $\mathcal{O}(\log(1/\kappa(n)))$ number of rounds, w.h.p.

\paragraph{Metastability.}
Consider the following worst case scenario:
Every time an \node in the core (periphery) 
chooses a random neighbor belonging to the periphery (core), 
then that neighbor is \emph{red} (\emph{blue}).
In this scenario, the \choice\ can be thought of
as two independent \advchoice{p}{$\sigma$}
in which for the core $p = q_{_\core}$ and $\sigma = \mathit{red}$,
and for the periphery $p = q_{_\periphery}$ and $\sigma = \mathit{blue}$.
From $c_r > \frac{1}{c^\star}$ and $c_d < c^\star$ and
Equations~\eqref{eq:robustnessprobability_rel} and~\eqref{eq:dominanceprobability_rel},
it follows that $q_{_\core}$ and $q_{_\periphery}$ are less than $\D$.
By applying the metastability result of Theorem~\ref{thm:noisyd},
we get that the volume of the adversary never exceeds $\frac{1-3p}{4(1-p)}$
of the network's volume.
Since both $q_{_\core}$ and $q_{_\periphery}$ are smaller than $\D$,
we have that $\frac{1-3p}{4(1-p)} \leq \frac{1}{4}$
(as the inequality is tight for $p=0$ and its value is decreasing).
Thus, the volumes of \emph{red} (\emph{blue}) agents in the core (periphery) 
are at most a fraction of $\frac{1}{4}$.
Therefore, the volumes of \emph{blue} and \emph{red} agents in the whole network 
are at least $\frac{3}{4}$ of the volumes of $\core$ and $\periphery$, respectively.
\end{proof}

\subsection{Experimental Results}
\label{sec:experiments}

In the previous Section~\ref{sec:theo} we formally studied the \choice\ 
on Core-Periphery networks, observing a phase transition phenomenon
that appears on a \emph{dominance} threshold $c^\star = \frac{\sqrt{2}-1}{2}$. 
Here, we report the results of the \emp data 
obtained by simulating the \choice\ on real-world networks.
Furthermore, we discuss our results 
and compare them with our theoretical analysis.
The source code of the experiments is freely available.%
\footnote{\url{https://gitlab.com/anusser/2_choices_on_core_periphery/-/tags/AAMAS2018}}

We simulated the \choice\ on 70 real-world networks, 
out of which 25 taken from KONECT~\cite{kunegis2013konect} 
and 45 from SNAP~\cite{snapnets}.
Detailed information regarding the networks and the results of the experiments 
are reported in Table~\ref{tab:datasets}.
The networks chosen for the experiments are drawn from a variety of domains
including social networks, communication networks, road networks, and web graphs;
moreover, they range in size from thousands of nodes and thousands of edges
up to roughly one million of nodes and tens of millions of edges.
Before simulating the \choice{}, we pre-process the networks 
in order to match the theoretical setting.
In particular, for all the networks, we remove the orientation of the edges and all loops,
and we work on the largest (w.r.t.\ the number of nodes) connected component.

The first issue we faced simulating the \choice\ was
the extraction of the set of agents representing the core.
In fact, there is no exact definition of what a \emph{good} core is
with respect to \emph{dominance} and \emph{robustness} values.
We started by using a simple heuristic 
to extract the core, namely the \emph{$k$-rich-club} method~\cite{zhou2004rich}:
This method establishes the core $\core$ as the set of $k$ agents with highest degree
and the periphery $\periphery$ as the remaining agents.
Avin et al.~\cite{avin_core-periphery_2014} empirically show that
when $k$ is at the \emph{symmetry point},
i.e., $k$ is chosen such that $vol(\core) \approx vol(\periphery)$,
the core found by this method
is sublinear in size with respect to the number of agents of the network.
We remark that if $vol(\core) = vol(\periphery)$
then, from the definitions of \emph{robustness} and \emph{dominance}, 
it follows that
\(
	c_r = \frac{1}{c_d}.
\)

We initially used the \emph{$k$-rich-club} method to extract the core but
noted that this simple heuristic produces a core with very low
\emph{robustness} values, contrary to what common sense would suggest to be a
\emph{good core}.
In particular, low robustness values imply that the \emph{dominance} values never
go below our theoretical threshold $c^\star$ (see columns  $\bar{c}_r$ and $\bar{c}_d$ in Table~\ref{tab:datasets}), which hinders the comparison
between theoretical and experimental results.
Indeed, in our theoretical analysis we assume that the core never changes
color, i.e., that the \emph{robustness} is high; however, in the experiments the
core was very unstable when using the \emph{$k$-rich-club} method.
The main issue of this method is that it does not
take into account the topological structure of the network, e.g.,
if we consider a regular graph with a well defined core-periphery structure
(which satisfies Definition \ref{def:regcore})
the \emph{$k$-rich-club} method would identify the core to be a random subset of nodes.

\begin{algorithm}[ht!]
\caption{Densest-Core Extraction.}
\begin{algorithmic}[1]
\Procedure{DensestCore}{$G$}
	\State $\core^\star \gets \emptyset$
	\Do
		\State $\core \gets \emptyset;\, D \gets G$
		\While {$D \neq \emptyset$}
			\State $v \gets \Call{LowestDegreeNode}{D}$
			\State $D \gets D \setminus \{v\}$
			\If {$\Call{Density}{D} > \Call{Density}{\core}$ \textbf{and} \\
				$\Call{FractionOfVolume}{\core^\star \cup D} \leq \frac{1}{2}$}
				\State $\quad \core \gets D$
			\EndIf
		\EndWhile
		\State $\core^\star \gets \core^\star \cup \core$
		\State $G \gets G \setminus \core$
	\doWhile {$\core \neq \emptyset$}
	\State \Return $\core^\star$
\EndProcedure
\end{algorithmic}
\label{alg:densest-core}
\end{algorithm}

Therefore, we introduce a novel heuristic for extracting the core which takes the
network topology into account by
prioritizing the \emph{robustness} of the core 
over its \emph{dominance}.
The procedure, which we refer to as \emph{densest-core} method, 
is described in Algorithm~\ref{alg:densest-core}. 
Informally, it iteratively extracts the densest subgraph from the network
and adds it to the core unless the core's volume becomes too large.
In order to compute this constrained densest subgraph, it uses a variation of the 2-approximation algorithm
\cite{charikar2000greedy},
which chooses every time the densest subgraph
that will not make the core's volume larger than the periphery's volume. 

We apply the \emph{densest-core} method to the networks and, as expected, we
obtain higher \emph{robustness} and lower \emph{dominance} values 
compared to the \emph{$k$-rich-club} method. 
The data reported in Table~\ref{tab:datasets} shows that
the \emph{robustness} of the core extracted by our method is higher
in all the considered datasets but one. 
Indeed, we finally obtain \emph{dominance} values below
the theoretical threshold $c^\star$.

We proceed as follows:
We initialize all the agents in $\core$ with \emph{blue}
and all the agents in $\periphery$ with \emph{red}.
Then, we simulate the \choice{} on each network, keeping track of 
the volumes of \emph{blue} and \emph{red} agents in each iteration.
We declare an \emph{\consensus} on the majority's color 
if within $\abs{V}$ iterations either the \emph{red} or the \emph{blue} agents
reach a volume greater than 95\% of the network's volume.
Otherwise we consider the simulation \emph{metastable} --
waiting for a superpolynomial number of rounds would be infeasible.
The experiments were repeated 50 times for each network.

\begin{figure}[t!]
	\centering
	\includegraphics[width=\linewidth]{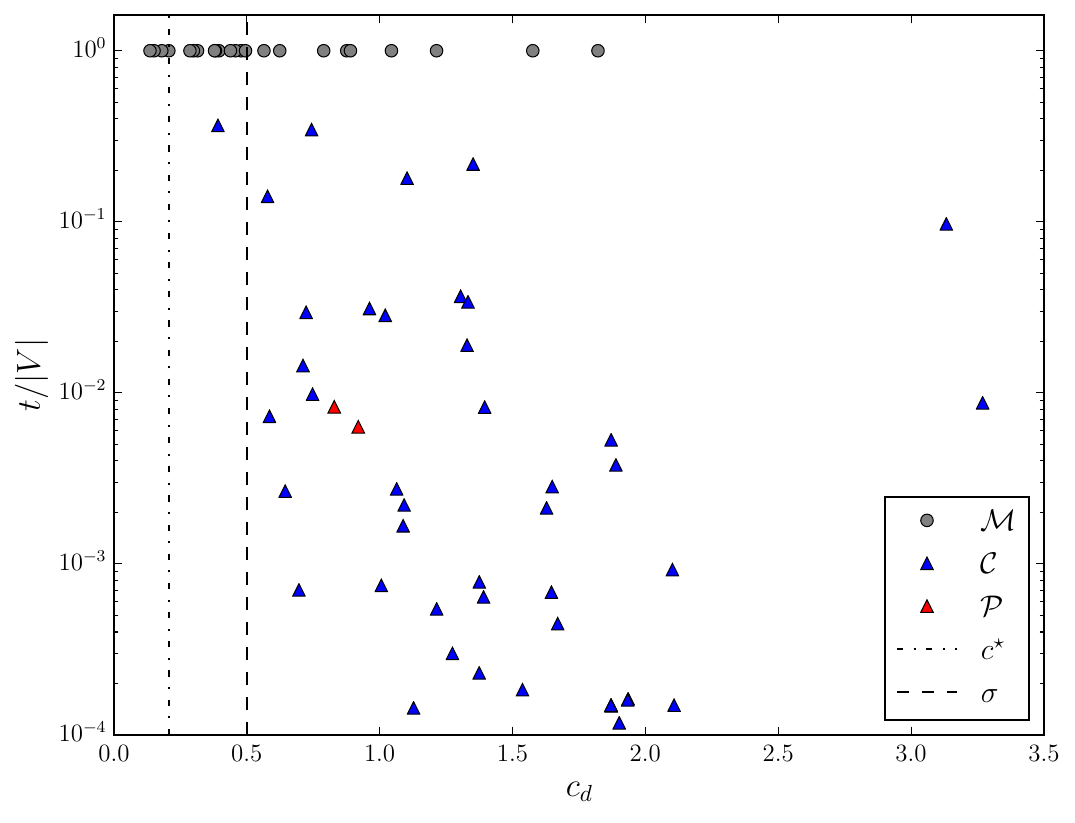}
	\caption{\Consensus\ and metastability of the experiments compared to 
	the theoretical and empirical thresholds $c^\star$ and $\sigma$.
	In 81\% of the runs there is an \consensus\ if $c_d > \sigma$;
	the 86\% of them are metastable when $c_d < \sigma$.
	The value $t$ is the arithmetic mean of the number of rounds 
	until \consensus/metastability was declared.
	}
	\label{fig:experimental-threshold}
\end{figure}

{\scriptsize
\begin{longtable}{l|rrccccc}
\caption{
Experimental data. \emph{Source} reports the source of the dataset, i.e.\ SNAP (S) or KONECT (K).
The values $c_r$ and $c_d$ are the \emph{robustness} and \emph{dominance}
obtained using the \emph{densest-core} method;
the values $\bar{c}_r$ and $\bar{c}_d$ are the \emph{robustness} and \emph{dominance}
obtained using the \emph{$k$-rich-club} method.
$\core$ and $\periphery$ are the fraction of experiments in which 
the core's and the periphery's color respectively spread to reach an \emph{\consensus}, 
while $\mathcal{M}$ is the fraction of experiments in which there is \emph{metastability},
all having the core extracted with the \emph{densest-core} method.
K stands for thousand, M for million.}
\\
\rowcolor{gray!30}
\emph{Dataset (Source)} & $|V|$ 	& $|E|$ 	& $c_r\,		(\bar{c}_r)$ 	& $c_d\,		(\bar{c}_d)$ 	& $\core$ & $\periphery$ & $\mathcal{M}$ \\
\hline
\endfirsthead
\rowcolor{gray!30}
\emph{Dataset (Source)} & $|V|$ 	& $|E|$ 	& $c_r\,		(\bar{c}_r)$ 	& $c_d\,		(\bar{c}_d)$ 	& $\core$ & $\periphery$ & $\mathcal{M}$ \\
\hline
\endhead
\endfoot
\endlastfoot
\rowcolor{gray!20}	Chicago (K)		& 0.8K		& 1.6K		& 6.55 		(0.10) & 0.15		(9.72) 	& 0.00 	& 0.00 	& 1.00 \\
\rowcolor{blue!20}	email-Eu-core (S)				& 0.9K		& 32.1K		& 0.75 		(0.53) & 1.32		(1.88) 	& 0.92 	& 0.08		& 0.00 \\
\rowcolor{gray!20}	Euroroad (K)			& 1.0K		& 2.6K		& 5.53 		(0.62) & 0.18		(1.61) 	& 0.00 	& 0.00 	& 1.00 \\
\rowcolor{gray!20}	Blogs (K)				& 1.2K		& 33.4K		& 0.62 		(0.38) & 1.57		(2.60) 	& 0.00 	& 0.00 	& 1.00 \\
\rowcolor{gray!20}	Traffic Control (K)					& 1.2K		& 4.8K		& 1.25 		(0.51) & 0.78		(1.96) 	& 0.00 	& 0.00 	& 1.00 \\
\rowcolor{blue!20}	Protein (K)				& 1.4K		& 3.8K		& 0.90 		(0.33) & 1.10		(2.95) 	& 1.00 	& 0.00 	& 0.00 \\
\rowcolor{gray!20}	US Airport (K)			& 1.5K		& 34.4K		& 0.54 		(0.48) & 1.82		(2.10) 	& 0.00 	& 0.00 	& 1.00 \\
\rowcolor{blue!20}	Stelzl (K)				& 1.6K		& 6.2K		& 1.03 		(0.36) & 0.96		(2.73) 	& 1.00 	& 0.00 	& 0.00 \\
\rowcolor{blue!20}	Bible (K)				& 1.7K		& 18.1K		& 0.74 		(0.54) & 1.33		(1.84) 	& 0.98 	& 0.02		& 0.00 \\
\rowcolor{blue!20}	Hamster full (K)				& 2.0K		& 32.1K		& 0.96 		(0.66) & 1.02		(1.51) 	& 1.00 	& 0.00 	& 0.00 \\
\rowcolor{blue!20}	Opsahl OF (K)			& 2.9K		& 31.2K		& 0.76 		(0.55) & 1.30		(1.81)		& 1.00 	& 0.00 	& 0.00 \\
\rowcolor{blue!20}	OpenFlights (K)					& 3.3K		& 38.4K		& 0.73 		(0.50) & 1.35		(1.98) 	& 0.80 	& 0.00 	& 0.20 \\
\rowcolor{blue!20}	bitcoin-alpha (S)				& 3.7K		& 28.2K		& 0.53 		(0.39) & 1.87		(2.52) 	& 1.00 	& 0.00 	& 0.00 \\
\rowcolor{gray!20}	ego-Facebook (S)		& 4.0K		& 176.4K	& 4.83 		(1.53) & 0.20		(0.65) 	& 0.00 	& 0.00 	& 1.00 \\
\rowcolor{gray!20}	ca-GrQc (S)						& 4.1K		& 26.8K		& 3.33 		(1.29) & 0.29		(0.77) 	& 0.00 	& 0.00 	& 1.00 \\
\rowcolor{gray!20}	US power grid (K)			& 4.9K		& 13.1K		& 3.17 		(0.53) & 0.31		(1.86) 	& 0.00 	& 0.00 	& 1.00 \\
\rowcolor{blue!20}	bitcoin-otc (S)					& 5.8K		& 42.9K		& 0.52 		(0.38) & 1.88		(2.59) 	& 1.00 	& 0.00 	& 0.00 \\
\rowcolor{red!20}	p2p-Gnutella08 (S)				& 6.2K		& 41.5K		& 1.20 		(0.53) & 0.82		(1.86) 	& 0.00 	& 1.00 	& 0.00 \\
\rowcolor{blue!20}	Route Views (K)					& 6.4K		& 25.1K		& 0.30 		(0.16) & 3.26		(6.13) 	& 0.96 	& 0.04		& 0.00 \\
\rowcolor{blue!20}	wiki-Vote (S)					& 7.0K		& 201.4K	& 0.60 		(0.44) & 1.64		(2.24) 	& 1.00 	& 0.00 	& 0.00 \\
\rowcolor{red!20}	p2p-Gnutella09 (S)				& 8.1K		& 52.0K		& 1.08 		(0.53) & 0.91		(1.86) 	& 0.00 	& 1.00 	& 0.00 \\
\rowcolor{blue!20}	ca-HepPh (S)					& 8.6K		& 49.6K		& 1.40 		(0.69) & 0.71		(1.44) 	& 1.00 	& 0.00 	& 0.00 \\
\rowcolor{blue!20}	p2p-Gnutella06 (S)				& 8.7K		& 63.0K		& 0.91 		(0.53) & 1.09		(1.87) 	& 1.00 	& 0.00 	& 0.00 \\
\rowcolor{blue!20}	p2p-Gnutella05 (S)				& 8.8K		& 63.6K		& 0.93 		(0.54) & 1.06		(1.83) 	& 0.86 	& 0.14		& 0.00 \\
\rowcolor{blue!20}	PGP (K)					& 10.6K		& 48.6K		& 2.54 		(1.18) & 0.39		(0.84) 	& 1.00 	& 0.00 	& 0.00 \\
\rowcolor{blue!20}	p2p-Gnutella04 (S)				& 10.8K		& 79.9K		& 0.91 		(0.52) & 1.08		(1.90) 	& 1.00 	& 0.00 	& 0.00 \\
\rowcolor{gray!20}	ca-HepTh (S)					& 11.2K		& 235.2K	& 3.49 		(2.39) & 0.28		(0.41) 	& 0.00 	& 0.00 	& 1.00 \\
\rowcolor{blue!20}	ca-AstroPh (S)					& 17.9K		& 393.9K	& 1.54 		(0.84) & 0.64		(1.18) 	& 1.00 	& 0.00 	& 0.00 \\
\rowcolor{blue!20}	ca-CondMat (S)					& 21.3K		& 182.5K	& 1.70 		(0.68) & 0.58		(1.46) 	& 1.00 	& 0.00 	& 0.00 \\
\rowcolor{blue!20}	p2p-Gnutella25 (S)				& 22.6K		& 109.3K	& 0.72 		(0.41) & 1.37		(2.43) 	& 1.00 	& 0.00 	& 0.00 \\
\rowcolor{blue!20}	E.A.T. (K)							& 23.1K		& 594.1K	& 0.60 		(0.48) & 1.64		(2.07) 	& 0.96 	& 0.04		& 0.00 \\
\rowcolor{blue!20}	Cora citation (K)				& 23.1K		& 178.3K	& 1.37 		(0.54) & 0.72		(1.83) 	& 1.00 	& 0.00 	& 0.00 \\
\rowcolor{blue!20}	CAIDA (K)			& 26.4K		& 106.7K	& 0.31 		(0.16) & 3.13		(6.03) 	& 1.00 	& 0.00 	& 0.00 \\
\rowcolor{blue!20}	p2p-Gnutella24 (S)				& 26.4K		& 130.7K	& 0.71 		(0.42) & 1.39		(2.34) 	& 1.00 	& 0.00 	& 0.00 \\
\rowcolor{blue!20}	cit-HepTh (S)					& 27.4K		& 704.0K	& 1.33 		(0.74) & 0.74		(1.34) 	& 1.00 	& 0.00 	& 0.00 \\
\rowcolor{blue!20}	Digg (K)			& 29.6K		& 169.5K	& 0.59 		(0.49) & 1.67		(2.01) 	& 1.00 	& 0.00 	& 0.00 \\
\rowcolor{blue!20}	Linux (K)						& 30.8K		& 426.4K	& 0.47 		(0.24) & 2.10		(4.14) 	& 0.90 	& 0.10 	& 0.00 \\
\rowcolor{blue!20}	email-Enron (S)					& 33.6K		& 361.6K	& 0.71 		(0.54) & 1.39		(1.84) 	& 1.00 	& 0.00 	& 0.00 \\
\rowcolor{blue!20}	cit-HepPh (S)					& 34.4K		& 841.5K	& 1.34 		(0.61) & 0.74		(1.61) 	& 1.00 	& 0.00 	& 0.00 \\
\rowcolor{blue!20}	Internet topology (K)					& 34.7K		& 215.4K	& 0.61 		(0.32) & 1.62		(3.08) 	& 0.88 	& 0.00 	& 0.12 \\
\rowcolor{blue!20}	p2p-Gnutella30 (S)				& 36.6K		& 176.6K	& 0.82 		(0.44) & 1.21		(2.23) 	& 1.00 	& 0.00 	& 0.00 \\
\rowcolor{blue!20}	loc-Brightkite (S)				& 56.7K		& 425.8K	& 0.99 		(0.71) & 1.00		(1.40) 	& 1.00 	& 0.00 	& 0.00 \\
\rowcolor{blue!20}	p2p-Gnutella31 (S)				& 62.5K		& 295.7K	& 0.78 		(0.44) & 1.27		(2.26) 	& 1.00 	& 0.00 	& 0.00 \\
\rowcolor{blue!20}	soc-Epinions1 (S)				& 75.8K		& 811.4K	& 0.72 		(0.60) & 1.37		(1.65) 	& 1.00 	& 0.00 	& 0.00 \\
\rowcolor{blue!20}	Slashdot081106 (S)		& 77.3K		& 937.1K	& 0.51 		(0.46) & 1.93		(2.13) 	& 0.98 	& 0.02	& 0.00 \\
\rowcolor{blue!20}	soc-Slashdot0811 (S)			& 77.3K		& 938.3K	& 0.51 		(0.46) & 1.93		(2.13) 	& 1.00 	& 0.00 	& 0.00 \\
\rowcolor{gray!20}	ego-Twitter (S)		& 81.3K		& 2.6M	& 1.12 		(0.57) & 0.89		(1.75) 	& 0.00 	& 0.00 	& 1.00 \\
\rowcolor{blue!20}	Slashdot090216 (S)		& 81.8K		& 995.3K	& 0.53 		(0.48) & 1.87		(2.08) 	& 1.00 	& 0.00 	& 0.00 \\
\rowcolor{blue!20}	Slashdot090221 (S)		& 82.1K		& 1.0M	& 0.53 		(0.48) & 1.87		(2.08) 	& 1.00 	& 0.00 	& 0.00 \\
\rowcolor{blue!20}	soc-Slashdot0922 (S)			& 82.1K		& 1.0M	& 0.53 		(0.47) & 1.87		(2.08) 	& 1.00 	& 0.00 	& 0.00 \\
\rowcolor{gray!20}	Prosper loans (K)				& 89.1K		& 6.6M	& 0.82 		(0.47) & 1.21		(2.10) 	& 0.00 	& 0.00 	& 1.00 \\
\rowcolor{blue!20}	Livemocha (K)					& 104.1K	& 4.3M	& 0.47 		(0.38) & 2.10		(2.56) 	& 0.94 	& 0.06		& 0.00 \\
\rowcolor{gray!20}	Flickr (K)					& 105.7K	& 4.6M	& 2.27 		(1.07) & 0.43		(0.92) 	& 0.00 	& 0.00 	& 1.00 \\
\rowcolor{gray!20}	ego-Gplus (S)					& 107.6K	& 24.4M	& 0.95 		(0.54) & 1.04		(1.82) 	& 0.00 	& 0.00 	& 1.00 \\
\rowcolor{blue!20}	epinions (S)			& 119.1K	& 1.4M	& 0.64 		(0.52) & 1.53		(1.89) 	& 1.00 	& 0.00 	& 0.00 \\
\rowcolor{blue!20}	Github (K)						& 120.8K	& 879.7K	& 0.88 		(0.70) & 1.12		(1.41) 	& 1.00 	& 0.00 	& 0.00 \\
\rowcolor{blue!20}	Bookcrossing (K)		& 185.9K	& 867.2K	& 0.52 		(0.34) & 1.90		(2.87) 	& 1.00 	& 0.00 	& 0.00 \\
\rowcolor{gray!20}	loc-Gowalla (S)					& 196.5K	& 1.9M	& 1.14 		(0.80) & 0.87		(1.24) 	& 0.02 	& 0.00 	& 0.98 \\
\rowcolor{gray!20}	email-EuAll (S)					& 224.8K	& 679.8K	& 0.16 		(0.06) & 6.19		(14.4) 	& 0.00 	& 0.00 	& 1.00 \\
\rowcolor{gray!20}	web-Stanford (S)				& 255.2K	& 3.8M	& 2.52 		(0.37) & 0.39		(2.68) 	& 0.00 	& 0.00 	& 1.00 \\
\rowcolor{blue!20}	amazon0302 (S)					& 262.1K	& 1.7M	& 2.61 		(0.44) & 0.38		(2.23) 	& 1.00 	& 0.00 	& 0.00 \\
\rowcolor{blue!20}	com-DBLP (S)					& 317.0K	& 2.0M	& 1.43 		(0.70) & 0.69		(1.42) 	& 1.00 	& 0.00 	& 0.00 \\
\rowcolor{gray!20}	web-NotreDame (S)				& 325.7K	& 2.1M	& 2.63 		(0.60) & 0.37		(1.65) 	& 0.00 	& 0.00 	& 1.00 \\
\rowcolor{blue!20}	com-amazon (S)					& 334.8K	& 1.8M	& 1.72 		(0.32) & 0.57		(3.03) 	& 0.98 	& 0.00 	& 0.02 \\
\rowcolor{gray!20}	amazon0312 (S)					& 400.7K	& 4.6M	& 2.18 		(0.41) & 0.45		(2.41) 	& 0.00 	& 0.00 	& 1.00 \\
\rowcolor{gray!20}	amazon0601 (S)					& 403.3K	& 4.8M	& 2.08 		(0.40) & 0.47		(2.44) 	& 0.00 	& 0.00 	& 1.00 \\
\rowcolor{gray!20}	amazon0505 (S)					& 410.2K	& 4.8M	& 2.01 		(0.41) & 0.49		(2.40) 	& 0.00 	& 0.00 	& 1.00 \\
\rowcolor{gray!20}	web-BerkStan (S)				& 654.7K	& 13.1M	& 1.60 		(0.43) & 0.62		(2.30) 	& 0.00 	& 0.00 	& 1.00 \\
\rowcolor{gray!20}	web-Google (S)					& 855.8K	& 8.5M	& 1.77 		(0.42) & 0.56		(2.33) 	& 0.00 	& 0.00 	& 1.00 \\
\rowcolor{gray!20}	roadNet-PA (S)					& 1.0M	& 3.0M	& 7.35 		(1.01) & 0.13		(0.98) 	& 0.00 	& 0.00 	& 1.00
\label{tab:datasets}\\
\end{longtable}
}

As can be observed in Figure~\ref{fig:experimental-threshold},
there exists an \emph{empirical threshold}~$\sigma = \frac{1}{2}$
which is different from the theoretical one.
In fact, in $81\%$ of the datasets with a \emph{dominance} above the threshold
the \choice\ converges to an \emph{\consensus}
while in $86\%$ of the datasets with a \emph{dominance} below the threshold,
the \choice\ ends up in a \emph{metastable} phase.
The \emp threshold is greater than the theoretical threshold because of several factors:
(i) in the experiments the core actually changes color to a small extent (while in the
theoretical part we ignored such small \emph{perturbations}),
and it consequently lowers the probability for an agent in the periphery
to pick the core's color;
(ii) the real-world network we used in the experiments do not have the regularity assumptions of the networks that we consider in the analysis;
(iii) in the experiments we declare metastability only after $|V|$ iterations and this increases the likelihood of false positive metastable runs.
The gap between the theoretical and the empirical threshold should be closed in
future work by providing a more fine-grained theoretical analysis which does not
assume the adversary's color to be monochromatic and considers more general networks.

We want to highlight that the protocol's convergence to the core's color (as shown in
Table~\ref{tab:datasets}) is remarkable in light of the fact that 
the \emph{densest-core} method ensures that the sum of the agents' degrees 
of the core and of the periphery are equal.
More precisely, note that equal volumes of core and periphery,
starting from an initial configuration where two sets support different colors,
is sufficient in the Voter Model to say that the two initial colors have 
the same probability to be the one eventually supported by all 
\nodes{}~\cite{hassin_distributed_2001}, regardless of the topological
structure. 
Previous works on the \choice{}~\cite{cooper_power_2014} provided convergence
results which are parametrized only in the difference of the volumes of the two
sets, suggesting a similar behavior. Our experimental results highlight the
insufficiency of the initial volume distribution as an accurate predictive
parameter, showing that the topological structure of the core plays a decisive
role.

\section{Conclusions}
We analyzed the \choice\ on a class of networks with core-periphery structure,
where the \emph{core}, a small group of densely interconnected agents,
initially holds a different opinion 
from the rest of the network, the \emph{periphery}.
We formally proved that a \emph{phase-transition} phenomenon occurs:
Depending on the dominance parameter $c_d$ characterizing the connectivity of the network,
either the core's opinion spreads among the agents of the periphery 
and the network reaches an \emph{(almost-)consensus},
or there is a \emph{metastability} phase 
in which none of the opinions prevails over the other.

We validated our theoretical results on several real-world networks.
Introducing an efficient and effective method to extract the core,
we showed that the same parameter $c_d$ is sufficient to predict 
the convergence/metastability of the \choice\ most of the time.
Surprisingly, even if the volumes of core and periphery are equal, 
the core's opinion wins in most of the cases.
These behaviors suggest that in many real-world networks 
there actually is a core whose initial opinion
has a great advantage of spreading in simple opinion dynamics such as the 2-Choices.
We think that these results are a relevant step towards understanding which dynamical properties are implicitly responsible for causing social and economic agents to form networks with a core-periphery structure.

\bibliographystyle{alpha}
\bibliography{biblio}

\end{document}